\documentclass[twocolumn,english,aps,pra,longbibliography,superscriptaddress, floatfix]{revtex4-1}

\usepackage[utf8]{inputenc}
\usepackage{mathtools,amsthm}
\usepackage{amssymb}
\usepackage{hyperref}
\usepackage{xcolor,soul}
\usepackage{boldline, multirow}
\usepackage{comment}
\usepackage{tabularx,booktabs}
\usepackage{graphicx}

\theoremstyle{plain}
\newtheorem{theorem}{Theorem}
\newtheorem{lemma}[theorem]{Lemma}
\newtheorem{proposition}[theorem]{Proposition}
\newtheorem{corollary}[theorem]{Corollary}
\theoremstyle{definition}

\theoremstyle{remark}

\newtheorem{remark}[theorem]{Remark}

\newcommand{\expec}[1]{\mathbb{E}[#1]}
\newcommand\numberthis{\addtocounter{equation}{1}\tag{\theequation}}
\newcommand{\norm}[1]{\left\lVert#1\right\rVert}
\newcommand{\calX}{\mathcal{X}}
\newcommand{\calB}{\mathcal{B}}

\newcommand{\SO}{\mathcal{S}}

\DeclareMathOperator*{\argmin}{arg\,min}
\DeclareMathOperator{\vol}{vol}
\DeclareMathOperator{\dist}{dist}
\DeclareMathOperator{\LMI}{LMI}
\DeclareMathOperator{\eig}{eig}

\DeclareMathOperator{\rand}{rand}
\DeclareMathOperator{\eye}{eye}
\DeclareMathOperator{\blkdiag}{blkdiag}
\DeclareMathOperator{\triu}{triu}
\DeclareMathOperator{\diag}{diag}
\DeclareMathOperator{\cg}{cg}
\DeclarePairedDelimiter{\ceil}{\lceil}{\rceil}

\def\norm#1{\mathopen\| #1 \mathclose\|}

\def\vol{\mbox{vol}}

\newcommand{\R}{\mathbb{R}}

\newcommand{\nz}{\textrm{nz}}

\begin{document}

\title{A Cutting-plane Method for Semidefinite Programming \\ with Potential Applications on Noisy Quantum Devices}

\author{J. Mare\v{c}ek}
\email{jakub.marecek@fel.cvut.cz}
\affiliation{Department of Computer Science, Czech Technical University, Karlovo nam.~13, Prague 2, the Czech Republic}
\author{A. Akhriev}
\email{albert\_akhriev@ie.ibm.com}
\affiliation{IBM Quantum, IBM Research -- Europe, IBM Dublin Technology Campus, Mulhuddart Dublin 15, Ireland}

\begin{abstract}
    There is an increasing interest in quantum algorithms for optimization problems. Within convex optimization, interior-point methods and other recently proposed quantum algorithms are non-trivial to implement on noisy quantum devices. Here, we discuss how to utilize an alternative approach to convex optimization, in general, and semidefinite programming (SDP), in particular. This approach is based on a randomized variant of the cutting-plane method. We show how to leverage quantum speed-up of an eigensolver in speeding up an SDP solver utilizing the cutting-plane method. 
    For the first time, we demonstrate a practical implementation of a randomized variant of the cutting-plane method for semidefinite programming on instances from SDPLIB, a well-known benchmark.
 Furthermore, we show that the RCP method is very robust to noise in the boundary oracle, which may make RCP suitable for use even on noisy quantum devices.     
    \end{abstract}

\maketitle

\section{Introduction}

Considering that gate-based quantum computers are expected to aid in solving specific optimization problems across many domains, including quantum chemistry  \cite{Ganzhorn2019}, machine learning \cite{Havlicek2019}, and computational finance \cite{egger2020quantum}, it may seem natural to seek quantum algorithms for convex  optimization.
Quantum speedup in convex optimization seems elusive, in general. Garg \emph{et al.} \cite{garg2021} have shown that in optimizing a Lipschitz-continuous, but otherwise arbitrary convex function over the unit ball,  first-order methods \cite{garg2021} have no quantum speedup over gradient descent, when restricted to the black-box access to the values and gradients of the convex function.
One should hence consider other special cases, preferably as broad as possible, and possibly avoiding the black-box access. 

Semidefinite programming (SDP) is a broad special case of convex optimization, which has attracted a substantial recent interest.
Initially, \cite{8104077,8104076,brandao2019quantum,van2019improvements}
``quantized'' the so-called multiplicative-weight-update (MWU) algorithm
of Arora and Kale \cite{arora2007combinatorial} and its variants by Hazan \cite{hazan2008sparse}. 
\footnote{We refer to Algorithm 6 in \cite{8104077} for a nice overview of the algorithm.}
Subsequently, \cite{kerenidis2018quantum,augustino2021inexact} attempted a translation of primal-dual interior-point methods  \cite{wright1997primal} to quantum computers.
\footnote{
We refer to Chapter 1 of \cite{wright1997primal} for an excellent introduction to primal-dual interior-point methods.
Due to the reliance of reliance of interior-point methods on solving linear systems, \cite{kerenidis2018quantum,augustino2021inexact} ended up with a bound dependent on the condition number $\kappa$ of a linear system based on the Karush-Kuhn-Tucker (KKT) conditions. As is well known \cite[p. 215]{wright1997primal}, this goes to infinity for all instances, by the design of the method, which may be not ideal in practice.
Furthermore, there is the issue of the HHL algorithm \cite{harrow2009quantum} providing the solution of the linear system only as a quantum state, whereas the interior-point method \cite{kerenidis2018quantum,augustino2021inexact} needs a classical update. The HHL hence needs to be run many times and the quantum state measured many times, to estimate the classical update.}
Finally, in~\cite{gilyen2019optimizing,van2020quantum,chakrabarti2020quantum}, the authors study the relationship of several oracles useful in first-order algorithms, but do not claim a run-time of a particular algorithm for SDPs.
These results are summarized in Table~\ref{tab:cvxopt}. \footnote{As it has been shown in \cite[Appendix E]{van2020quantum}, in the MWU algorithm, $\frac{rR}{\epsilon}$
should be seen as an important parameter, as one can trade-off dependence on one of the
three individual parameters for the dependence on the others.}
While some of the quantum algorithms \cite{brandao2019quantum} are reported as scaling with $O(\sqrt{m} \ \textrm{poly}(\log(m), \log n))$ \cite{brandao2019quantum} for $m$ constraints in $n \times n$ matrices,
this requires the diameter of the convex set to be independent of the dimension, while the dependence is quadratic, in general. Furthermore, none of these algorithms have been implemented in an actual quantum device, or its simulator. 

Here, we consider another method for solving SDPs, which can be run in part on the quantum computer. In particular, we ``quantize'' the so-called randomized cutting plane (RCP) method.
The cutting-plane methods \cite{grotschel2012geometric} have produced a variety of classical theoretical guarantees, as surveyed in 
Table~\ref{tab:cutting_plane_method},  including the first polynomial-time algorithm for linear programming, but yielded little in terms of practical  implementations in classical computers.
This is because a certain sub-routine, known as the boundary oracle, is classically almost as demanding as the original problem.
We show how to leverage quantum speed-up of an eigensolver in speeding up the RCP method. 
Furthermore, we demonstrate an implementation of the RCP method, which is very robust to noise in the boundary oracle. The robustness to noise may make RCP suitable for use even on noisy quantum devices, which are available within the foreseeable future. 

We formalize the problem and discuss the related work in more detail in Sec.~\ref{sec:prelim}.
We present our main result in Sec.~\ref{sec:main}.
Finally, we discuss our numerical results in Sec.~\ref{sec:results}.



\begin{table*}[t!]
\caption{An overview of recently proposed quantum algorithms for convex optimization on quantum computers, sorted by their appearance in arxiv (listed under Year). In the upper bounds, we drop polylogarithmic
terms.}
\label{tab:cvxopt}
\begin{tabularx}{\textwidth}{@{\extracolsep{\fill}}lllXl}
\hline
    Reference & Year & Algorithm  & Complexity & Complexity ref. \\ \hline
\cite{8104077} & 2016 & Multiplicative weights update &  $\sqrt{mn}s^2 (Rr/\epsilon)^{32}$ & Cor. 17 \\
\cite{8104076,van2020quantum,van2020phd} & 2017 & Multiplicative weights update &  $\sqrt{mn}s^2 (Rr/\epsilon)^{8}$ & Thm. 1 of \cite{van2020quantum} \\ 
\cite{brandao2019quantum} & 2017 & Multiplicative weights update & $(\sqrt{m}/\epsilon^{10} +\sqrt{n}/\epsilon^{12} ) s^2 \textrm{poly}(Rr/\epsilon)$ & Cor. 6 
\\
\cite{van2019improvements,van2020phd} & 2018 & Multiplicative weights update & $((\sqrt{m} + \sqrt{n}(Rr/\epsilon)) s(Rr/\epsilon)^{4} \approx Rrs(\sqrt{m}/\epsilon^4 + \sqrt{n}/\epsilon^5) $ & Thm. 17 of \cite{van2019improvements} \\ 
\cite{kerenidis2018quantum} & 2018 & Interior-point method &  $(n^{\omega} + \kappa n^2/\delta^2) \log{1/\epsilon}$, with $\kappa \to \infty$ & Cor 7.7  \\
\cite{van2020oracles} & 2018 & Subgradient & Not given &  \\
\cite{chakrabarti2020quantum} & 2018 & Subgradient & Not given &  \\
\cite{mohammadisiahroudi2021efficient} & 2021 & Interior-point method & Not given \\
\cite{augustino2021inexact} & 2021 & Interior-point method & $(n^{2.5} / \epsilon + \kappa \delta^2)$ , with $\kappa \to \infty$ & Sec. 5.3 \\
\cite{bharti2021nisq} & 2021 & Hybrid, q. preprocessing & Not given & \\
\hline
\end{tabularx}
\end{table*}

\begin{table}[!t]
    \caption{ A short history of classical algorithms based on the cutting-plane method (focussing on the feasibility). In the upper bounds (listed under Complexity), we drop polylogarithmic
terms and let $\rho = nR/\epsilon$.}\label{tab:cutting_plane_method} 
\centering
    \begin{tabularx}{\columnwidth}{@{\extracolsep{\fill}}lllX}
    \hline
    Ref. & Year & Algorithm  & Complexity \\ \hline
     \cite{s77,yn76,k80} & 1979 & Ellipsoid method & $n^2  \SO  \log (\rho) + n^4 \log (\rho)$ \\ 
     \cite{kte88,nn89} & 1988 & Inscribed ellipsoid & $n \SO \log (\rho) + (n \log (\rho))^{4.5}$ \\  
     \cite{v89} & 1989 & Volumetric center & $n \SO \log (\rho) + n^{\omega+1} \log (\rho)$ \\ 
     \cite{av95} & 1995 & Analytic center & $n \SO \log^2 (\rho) + n^{\omega+1} \log^2 (\rho) + ( n \log (\rho) )^{ 2 + \omega / 2 } $ \\ 
     \cite{bv02} & 2004 & Random walk & $ n \SO \log (\rho) + n^7 \log (\rho)$ \\ 
     \cite{Calafiore2007} & 2007 & Random walk & probabilistic analyses \\
     \cite{4739188,dsp10} & 2010 & Random walk & probabilistic analyses \\ 
     \cite{lsw15} & 2015 & Hybrid center & $n \SO \log (\rho) + n^3 \log^3 (\rho)$ \\ 
    \cite{jiang2020improved} & 2019 & Volumetric center & $n \SO \log (\rho) + n^3 \log (\rho) $ \\ 
    \hline
    \end{tabularx}
\end{table}


\section{Preliminaries and Related Work}
\label{sec:prelim}

\subsection{Convex Optimization}

We consider a convex constrained optimization problem \cite{boyd2004convex} of the form
\begin{equation}\label{eq:1.1}
\begin{aligned}
\min_{x \in \mathcal{X}} c^T &x 
\end{aligned}
\end{equation}
where $\mathcal{X}\subset \mathbb{R}^n$ is a convex compact set with non-empty interior (``convex body'') and 
 $c \in \mathbb{R}^n$ defines a linear objective.
 The linear cost function is taken without loss of generality, since any convex constrained optimization problem can be reduced to this form,
 \footnote{Indeed, if $f(x)$ is convex,
\begin{equation*}
\begin{aligned}
 f(x) &\to \min\\
x &\in \mathcal{X},
\end{aligned}
\end{equation*}
by introducing a slack variable $t$, we obtain an equivalent optimization problem of the form (\ref{eq:1.1})
\begin{equation*}
\begin{aligned}
 & \min t\\
x &\in \mathcal{X},\, f(x) - t \leq 0
\end{aligned}
\end{equation*}}.
Furthermore, we assume that there exist two Euclidean balls $\mathcal{B}_r$ and $\mathcal{B}_R$ of radii $0<r<R$, such that $\mathcal{B}_r\subseteq\mathcal{X}\subseteq\mathcal{B}_R$. 

A particularly important class of constrained convex optimization problems are semidefinite programs \cite{anjos2011handbook}:
\begin{align}
\operatorname{inf}\ \langle C,X \rangle  
\textrm{ s.t. }  \mathcal{A}X=b, \tag{SDP} \label{eq:SDP}
X\succeq_{\mathcal{K}}0 
\end{align}
where cone $\mathcal{K}$ is the cone of positive semidefinite symmetric $n\times n$ matrices  $\mathcal{S}_{+}^{n}$, i.e., $\{X=X^{\intercal}\in\mathbb{R}^{n \times n} |\ X \text{ is positive semidefinite}\}$, and $\mathcal{A} \colon \mathcal{S}^{n} \to \mathbb{R}^{m}$ is a linear operator between $\mathcal{S}_{+}^{n}$ and $\mathbb{R}^{m}$:
\begin{align*}
  X &\mapsto \begin{pmatrix}
  \langle A_{1},X\rangle\\
  \dots\\
  \langle A_{m},X\rangle
  \end{pmatrix}.
\end{align*}
\footnote{
We still assume there are  two Euclidean balls $\mathcal{B}_r$ and $\mathcal{B}_R$ inscribed and outscribed to the feasible set
of \ref{eq:SDP}, which is known as the spectrahedron.
Radius $R$ of the ball $\mathcal{B}_R$ above can also be seen as an upper bound on the trace of an optimal primal solution of an SDP.
Notice that in the case of a general SDP, parameters $r$ and $R$ are not constants independent of dimension, but do grow with the dimension.}.
This is a proper generalization of linear programming (LP), second-order cone programming (SOCP), and convex cases of quadratically-constrained quadratic programming (QCQP), which underlie much of operations research.
SDPs also have extensive applications in combinatorial optimization, control engineering, (quantum) information theory, machine learning \cite{majumdar2020recent}, and statistics.
Under the Unique Games Conjecture \cite{khot2007optimal,KhotSurvey,khot2015unique}, randomized rounding  \cite{Raghavan1987} of SDPs obtains the best possible polynomial-time classical algorithms for a variety of problems.


Correspondingly, there have been proposed many classical algorithms for solving constrained convex optimization problems.
In summary, SDPs can be classically approximated to any precision in polynomial time.
Presently, both the best theoretical bounds \cite{jiang2020faster} and the best practical solvers \cite{mosek2020}
employ interior-point methods.
At least in theoretical models of computation \cite{blum2012complexity} where a real-number arithmetic operation can be performed in unit time,
there are classical upper bounds \cite{jiang2020faster} on the run time of the form  
${\mathcal{\tilde O}}(\sqrt{n}(mn^2+m^{\omega}+n^{\omega})\log(1/\epsilon))$,
where $\mathcal{O}(\cdot)$ indicates the Bachmann–Landau notation, 
tilde in $\mathcal{\tilde O}(\cdot)$ indicates that we drop the polylogarithmic terms, 
$n$ is the dimension of the problem, $\omega \in [2, 2.373)$ is the exponent for matrix multiplication, $m$ is the number of constraints, and  
$\nz$ is the maximal number of non-zero entries per row of the input matrices.
\footnotetext[1]{Notice that the analysis of  \cite{porkolab1997complexity} shows the situation is less trivial in the Turing machine, and one may need to consider the dimension or the number of constraints constant.}
For certain smooth instances with sufficient curvature, there are first-order methods \cite{yurtsever2019scalable}, which can be faster still.
Nevertheless, many of the instances of SDPs encountered in the real-world that are not possible to solve using classical computers in practice. 

To introduce the randomized cutting plane method, it is useful to consider perhaps the single most simple optimization algorithm possible: in each iteration, cut a convex body in two pars at its center of gravity, and repeat with the part that yields better objective function. See \textsc{Algorithm} $1$ for an outline.
This is known as the Deterministic Center-of-Gravity (DCG) algorithm and it has been proposed by Levine \cite{levin1965algorithm} and, independently, Newman \cite{newman1965location} in 1965.

\subsection{Deterministic Center-of-Gravity (DCG) Algorithm}
For a convex body $\calX \subset \mathbb{R}^n$, define its \emph{center of gravity} as 
$$\cg(\calX) = \frac{\int_{\calX}x\mathrm{d}x}{\int_{\calX}\mathrm{d}x}.$$

\begin{figure}[t]
    \begin{tabularx}{245pt}{l}
    \hline
     \textsc{Algorithm} $1$: DCG, cf. \cite{levin1965algorithm,newman1965location} \\
     \textbf{Input:} $\mathcal{X}$\\
     \textbf{Output:} $z_k$\\
       1: $k = 0$, $\mathcal{X}_k = \mathcal{X}$\\
       2: \textbf{repeat}: \\
       3: $z_k = \cg(\calX_k)$\\
       4: $\mathcal{X}_{k+1} = \{x\in\mathcal{X}_k:c^T(x-z_k)\leq 0\}$\\
       5: $k = k+1$ \\
       6: \textbf{until} a stopping criterion is satisfied.\\[0.3ex]
     \hline
    \end{tabularx}
\end{figure}

\begin{proposition}[Gr\"{u}nbaum \cite{grunbaum1960partitions}, \footnote{Note that some literature \cite{4739188,dsp10} restates the lemma incorrectly, with $\frac{1}{n}$ instead of $n$.}] Let $\mathcal{X} \subset \mathbb{R}^n$ be a convex body, and let $x_G = cg(\mathcal{X})$ be its center of gravity. Consider any hyperplane $\mathcal{H} = \{x\in\mathbb{R}^n: c^T (x-x_G) = 0\}$ passing through $x_G$. This hyperplane divides the set $\calX$ into two subsets
\begin{align*}
    \mathcal{X}_1 &= \{x\in \mathcal{X}: c^T x > c^T x_G\}, \\
    \mathcal{X}_2 &= \{x\in \mathcal{X}: c^T x \leq c^T x_G\}.
\end{align*}
Then, for $i=1,2$:
\begin{equation}
    \vol(\mathcal{X}_i) \leq \left(1-\left(\frac{n}{n+1}\right)^n\right)\vol(\mathcal{X}) \leq \left(1-\frac{1}{e}\right)\vol(\mathcal{X}).
\end{equation}
\end{proposition}
\begin{remark}
    Each step of the DCG algorithm guarantees that a given portion of the feasible set is cut out, i.e.,
\[
    \vol(\calX_{k+1}) \leq \left(1-\frac{1}{e}\right)\vol(\calX_k).
\]
Applying this inequality recursively, we obtain the volume inequality
\begin{equation}
    \vol(\calX_k) \leq \left(1-\frac{1}{e}\right)^k \vol(\calX_0) \approx (0.63)^k\vol(\calX_0),
\end{equation}
which proves that DCG has guaranteed geometric convergence in terms of \emph{volumes}.
\end{remark}

By applying Radon's theorem to the DCG algorithm,
 we obtain a reduction of \emph{cost-function values} at each step. Here:

\begin{proposition}[Radon, \cite{radon1921mengen}] Let $\calX \subset \mathbb{R}^n$ be a convex body and $x_G = cg(\mathcal{X})$ be its center of gravity. Denote by $\mathcal{H}$ an arbitrary $(n-1)$-dimensional hyperplane through $x_G$, and let $\mathcal{H}_1$ and $\mathcal{H}_2$ be the two hyperplane supporting $\calX$ and parallel to $\mathcal{H}$. Denote by
\[
    r(\mathcal{H}) = \frac{\min\{\dist(\mathcal{H}, \mathcal{H}_1), \dist(\mathcal{H}, \mathcal{H}_2)\}}{\max\{\dist(\mathcal{H}, \mathcal{H}_1), \dist(\mathcal{H}, \mathcal{H}_2)\}}
\]
the ratio of the distances from $\mathcal{H}$ to $\mathcal{H}_1$ and $\mathcal{H}_2$, respectively. Then
\[
    \min_{\mathcal{H}} r(\mathcal{H}) \geq \frac{1}{n}.
\]
\end{proposition}

Specifically:

\begin{equation}\label{eq:2.3}
    f_{k+1} - f^* \leq \frac{n}{n+1}(f_{k}-f^*)\tag{2.3}.
\end{equation}
Indeed, let
\begin{align*}
    \mathcal{H}_1 &= \{x\in\mathcal{X}: c^T x = c^T x_k\}\\
    \mathcal{H} &= \{x\in\mathcal{X}: c^T x = c^T x_{k+1}\}\\
    \mathcal{H}_2 &= \{x\in\mathcal{X}: c^T x = c^T x^*\},
\end{align*}
then
\begin{align*}
    \dist(\mathcal{H},\mathcal{H}_1) & = \frac{|c^T x_k - c^T x_{k+1}|}{\norm{c}} = \frac{c^T x_k - c^T x_{k+1}}{\norm{c}} \\ & = \frac{f_k - f_{k+1}}{\norm{c}}
\end{align*}
and
    $$\dist(\mathcal{H},\mathcal{H}_2) = \frac{f_{k+1} - f^*}{\norm{c}},$$
where we used the fact that $f_k \geq f_{k+1}$ and $f_{k+1} \geq f^*$.
By Radon's theorem,
$$\frac{f_k-f_{k+1}}{f_{k+1}-f^*} = \frac{\dist(\mathcal{H},\mathcal{H}_1)}{\dist(\mathcal{H},\mathcal{H}_2)} \geq r(\mathcal{H}) \geq  \min_{\mathcal{H}} r(\mathcal{H}) \geq \frac{1}{n},$$
$\Rightarrow$
$$\frac{f_k-f_{k+1}}{f_{k+1}-f^*}+1 \geq \frac{1}{n} + 1$$
$\Rightarrow$
$$f_{k+1} - f^* \leq \frac{n}{n+1}(f_{k}-f^*).$$

Thereby, we obtain the iteration complexity:

\begin{proposition}[Rate of convergence of DCG]\label{lem:2.1}
Define $D=f_0 - f^*$. Then the DCG algorithm computes an $\alpha$-optimal solution (i.e., such that $f_k - f^* \leq \alpha$) in a number of steps bounded as
\[
    k = \ceil*{\frac{\ln \frac{D}{\alpha}}{\ln \frac{n+1}{n}}} = \mathcal{O} \left(n\ln\frac{D}{\alpha}\right).
\]
\end{proposition}

This iteration complexity is essentially the same for all cutting-plane methods since 1988, as surveyed in Table~\ref{tab:cutting_plane_method}.
We refer to \cite{Bubeck2015} for an in-depth introduction.
Notice, however, that computing the deterministic center of gravity (``per-step complexity'') is \#P-hard even for 0-1 polytopes \cite{rademacher2007approximating}.
One would hence like to consider some alternative sub-routine, while preserving the same rate of convergence. 

\subsection{An RCP algorithm}

Over the past two decades, there have been developed cutting-plane algorithms \cite[e.g.]{bv02,dsp10} that replace the 
computing of the center of gravity of a convex body with 
sampling points uniformly at random from the convex body. 
An outline of such a randomized cutting-plane method is presented in \textsc{Algorithm} $2$.

\begin{table*}[t]
		\caption{An overview of geometric random walks: Their mixing times and upper bounds on the per-step complexity of sampling uniform distribution over the spectrahedron in the classical implementations by Chalkis et al. \cite{chalkis2020efficient}.}
	\centering
	\begin{tabularx}{510pt}{@{\extracolsep{\fill}}lllll}\hline
	\toprule
		Reference & Year & Random walk & Mixing time & Per-step complexity for \ref{eq:SDP}  \\ \hline 
	\cite{smith1984efficient,lovasz1999hit} & 1984 & 	Hit and Run 	 & Fast \cite{lovasz1999hit}  &  $\mathcal{O}(m^{2.697} + m \lg{L} + nm^2)$ \\
	\cite{smith1984efficient} & 1984 &	Coordinate-directions hit and run & Unknown	&  $\mathcal{O}(m^{2.697} + m \lg{L} + m^2)$ \\
	\cite{Polyak} & 2014 &	Billiard walk 	& Unknown & $\widetilde{\mathcal{O}}(\rho(m^{2.697} + m \lg{L} + nm^2))$ \\ 
	\cite{DUANE1987,Afshar2015,chevallier2018} & 2015 &	Hamiltonian Monte Carlo
with reflections \quad & Unknown & $\widetilde{\mathcal{O}}(\rho((dm)^{2.697} + md \lg{L} + dnm^2))$ \\ 
\cite{chevallier} & 2019 & Wang-Landau & Fast \cite{chevallier} & Unknown \\
\hline
	\end{tabularx}
		\label{tab:walks}
\end{table*}

\begin{figure}[ht]
    \begin{tabularx}{245pt}{l}
    \hline
     \textsc{Algorithm} $2$: Randomized cutting plane   \cite{bv02,calafiore2004random,4739188,polyak2006d}\\
     \textbf{Input:} $\mathcal{X}$\\
     \textbf{Output:} $z_k$\\
       1: $k = 0$, $\mathcal{X}_k = \mathcal{X}$\\
       2: \textbf{repeat}: \\
       3: generate $N_k$ uniformly distributed random samples \\ \quad \quad  in $\calX_k$, $\{x^{(1)},\ldots,x^{(N_k)}\}$, e.g., using \textsc{Algorithm} $3$ \\
       4: $z_k = \argmin_{x\in\{x^{(1)},\ldots,x^{(N_k)}\}}c^T x$\\
       5: $\mathcal{X}_{k+1} = \{x\in\mathcal{X}_k:c^T(x-z_k)\leq 0\}$\\
       6: $k = k+1$ \\
       7: \textbf{until} a stopping rule is satisfied.\\[0.3ex]
     \hline
    \end{tabularx}
\end{figure}

The uniform sampling is non-trivial, but a breakthrough result of \cite{lovasz2006fast} showed that it is possible using certain rapidly-mixing geometric random walks \cite{vempala2005geometric}. 
An overview of the geometric random walks is presented in Table~\ref{tab:walks}.
For any such random walk, one needs to provide one or more geometric subroutines, 
such as the 
test of membership of a point inside the set, a surface separating a point from the set,
etc. Several standard subroutines are beautifully surveyed in Chapter 3 of \cite{van2020phd}.
Our focus in this work will be on the so-called {\em Random Directions Hit and Run} random walk~\cite{smith1984efficient},
wherein the key subroutine is the intersection of a line (or curve, more generally) with the
boundary of the feasible set.
This subroutine is commonly known as the boundary oracle (\texttt{BO}).
See  \textsc{Algorithm} $3$ for an overview.

In the Supplementary Material, we present
 some background material concerning the statistical properties of the empirical minimum over a convex body 
 in Appendix \ref{app:statistical}.
 In Appendix \ref{sec:convergencerate}, we present an iteration complexity of the overall procedure, as captured in \textsc{Algorithm} $4$. In particular, we fix minor issues of previous analyses, especially those of Dabbene et al. \cite{4739188,dsp10}.
 In Appendix \ref{sec:impl}, we provide the full pseudo code of the algorithms, specialized to SDPs.
We note that the pseudocode and the bounds on the iteration complexity remain the same, independent of whether the boundary oracle is run classically or quantumly. 


\begin{figure}[t]
    \begin{tabularx}{245pt}{l}
    \hline
     \textsc{Algorithm} $3$: Hit-and-run (\texttt{H\&R}), cf. \cite{smith1984efficient,lovasz1999hit}\\
     \textbf{Input:} $\mathcal{X}$, $x_0 \in \calX$, $M$ (mixing time)\\
     \textbf{Output:} random point $x\in\calX$\\
       1: $y^{(0)} = x_0$\\
       2: \textbf{for} $i=0$ \textbf{to} $M-1$ \textbf{do}:\\
       3: generate a uniformly distributed random direction $v\in\mathbb{R}^n$\\
       4: $\{\underline{x}, \bar{x}\} = \texttt{BO}(\calX,y^{(i)},v)$\\
       5: generate a uniform point $y^{(i+1)}$ in the segment $[\underline{x}, \bar{x}]$\\
       6: \textbf{end for}\\
       7: $x = y^{(M)}$\\[0.3ex]
     \hline
    \end{tabularx}
\end{figure}

\begin{figure}[t]
    \begin{tabularx}{245pt}{l}
    \hline
     \textsc{Algorithm} $4$: RCP with \texttt{H\&R}, cf. \cite{calafiore2004random,4739188,polyak2006d}\\
     \textbf{Input:} $\mathcal{X}$, $x_0 \in \calX$, $M$\\
     \textbf{Output:} $z_k$\\
       1: $k = 0$, $\mathcal{X}_k = \mathcal{X}$\\
       2: \textbf{repeat}:\\
       3: \textbf{for} $j=1$ \textbf{to} $N$ \textbf{do}: $x^{(j)} = \texttt{H\&R}(x^{(j-1)},M)$; \textbf{end for}\\
       4: $z_k = \min_{x\in\{x^{(1)},\ldots,x^{(N)}\}}c^T x$\\
       5: $\mathcal{X}_{k+1} = \{x\in\mathcal{X}_k:c^T(x-z_k)\leq 0\}$\\
       6: $k = k+1$\\
       7: \textbf{until} a stopping rule is satisfied.\\[0.3ex]
     \hline
    \end{tabularx}
\end{figure}



\subsection{Boundary Oracle for Hit-and-Run Walks on the Feasible Set of an SDP}

Let us now consider the complexity of implementing a boundary oracle
for the {\em Random Directions Hit and Run} random walk~\cite{smith1984efficient}
for sampling uniformly at random from the spectrahedron \eqref{eq:SDP}.
For convenience, let us consider the dual of the semidefinite program \eqref{eq:SDP}, also known as the linear matrix inequality (LMI):
\begin{align}\label{eq:7.1}
\min c^T x \textrm{ s.t. }
F(x) &= F_0 + \sum_{i=1}^n x_i F_i \preceq 0,  \tag{LMI}
\end{align}
where $c\in\mathbb{R}^n$ and $F_i=F_i^T\in\mathbb{R}^{m\times m}$, $i=0,\ldots,n$, are known symmetric matrices. We then have the convex set
$$\calX = \calX_{\LMI} = \{x\in\mathbb{R}^n: F(x)\preceq 0\}.$$
We assume $\calX_{\LMI}$ is nonempty and bounded.

Given $y\in\calX_{\LMI}$ such that $F(y) \prec 0$ and a random direction $v\in \mathbb{R}^n$, how do we find the intersection points of the line $z=y+\lambda v$ and the boundary of $\calX_k$ at the $k$-th iteration? First, we have
$$F(y+\lambda v) = F(y) + \lambda(F(v)-F_0)\triangleq A+\lambda B,\ \lambda\in\mathbb{R}.$$
Next, we obtain the intersection points with the boundary of $\calX_{\LMI}$: $\underline{z}=y+\underline{\lambda}v$ and $\bar{z} = y+\bar{\lambda}v$ and test if $\underline{z},\,\bar{z}\in \calX_k = \{x\in\calX_{k-1}:c^T(x-z_{k-1})\leq 0\}$. If both points are in $\calX_k$, then $\{\underline{z},\,\bar{z}\}$ are the intersection points we need. Otherwise, only one of them  $\notin \calX_k$, so w.l.o.g. assume $\bar{z}\notin \calX_k$, and we need to find the intersection point between the line $z=y+\lambda v$ and the hyperplane $\{x\in\calX:c^T(x-z_{k-1}) = 0\}$, which can be easily obtained by solving for $\lambda$ in $c^T(y+\lambda v - z_{k-1}) = 0$. Let $\lambda^*$ denote the solution and let $\bar{z}' = y+\lambda^* v$. Then, $\{\underline{z},\bar{z}'\}$ are the desired intersection points. 

The work of \cite{calafiore2004random,4739188,polyak2006d} can be summarized as follows:

\begin{lemma}[Boundary oracle for LMIs, Lemma 6 in \cite{4739188}]\label{lem:7.1}
Let $A\prec 0$ and $B=B^T$. Then, the minimal and the maximal values of the parameters $\lambda \in \mathbb{R}$ retaining the negative definiteness of the matrix $A+\lambda B$ are given by
\[
\underline{\lambda} = \left\{\begin{array}{lr}
     \max\limits_{\lambda_i<0} \lambda_i&\\
     -\infty &\text{if all}\ \lambda_i > 0
\end{array}\right.
\]
and
\[
\overline{\lambda} = \left\{\begin{array}{lr}
     \min\limits_{\lambda_i>0} \lambda_i&\\
     +\infty &\text{if all}\ \lambda_i < 0,
\end{array}\right.
\]
where $\lambda_i$ are the generalized eigenvalues of the pair of matrices $(A,-B)$, i.e., $Av_i=-\lambda_iBv_i$.
\end{lemma}

The semidefinite generalized eigenvalue problem \cite[Chapter 3]{Lucas2004} could be seen as a special cases of the polynomial eigenvalue problem \cite{coise2000backward,Guttel2007}. There, we wish to compute
$\lambda \in \R$ and $x \in \R^m$ satisfying
\begin{equation}
\label{eq:PEP} \tag{PEP}
(B_d \lambda^d + \cdots + B_1 \lambda + B_0) x = 0
\enspace ,
\end{equation}
where $B_i \in \R^{m \times m}$ are matrices, out of which 
$B_d$ and $B_0$ are invertible, and all could be seen as 
coefficients of a univariate matrix polynomial.

Despite much recent progress in computational approach
to the polynomial eigenvalue problem \cite{coise2000backward,berhanu2005polynomial,Armentano2019,beltran2019real},
and effective computational geometry for surfaces
\cite{boissonnat2006effective} more broadly,
a classical implementation of the boundary oracle 
that would make the hit-and-run walk on the feasible set of 
\ref{eq:SDP} (or \ref{eq:7.1}) is still lacking. In particular, the present best classical run-time bound is:

\begin{lemma}[Chalkis et al., \cite{chalkis2020efficient}]
	\label{lem:pep}
	Consider a \ref{eq:PEP} of degree $d$, involving matrices of dimension
	$m \times m$, with integer elements of bitsize at most $\tau$. 
	There is a randomized algorithm for computing the eigenvalues and the eigenvectors of \ref{eq:PEP} up to precision
	$\epsilon = 2^{-L}$, in time $\widetilde{\mathcal{O}}((md)^{\omega+3} (m d)^3 \tau)$.
    The arithmetic complexity  is
	$\widetilde{\mathcal{O}}( \delta^{2.697} + m d \log(1/\epsilon))$.
\end{lemma}

Unfortunately, this is not much easier than solving the original convex constrained optimization problem, which has  \cite{jiang2020faster} the 
arithmetic complexity
${\mathcal{\tilde O}}(\sqrt{n}(mn^2+m^{\omega}+n^{\omega})\log(1/\epsilon))$, and wherein important special cases \cite[e.g.]{van2020deterministic} can be solved in matrix-multiplication time. 


\section{A Boundary Oracle via Quantum Eigensolvers}
\label{sec:main}

Our main result 
is a family of quantum algorithms for the boundary oracle for hit-and-run walks on the feasible set of an SDP, or rather its dual \eqref{eq:7.1}.
Therein, we transform the generalized eigenvalue problem to an eigenvalue problem on a larger matrix, which makes it possible to use any quantum algorithm for computing  the eigenvalues of the larger matrix.
Quantum eigensolvers are, in turn, some of the best understood quantum algorithms \cite{kitaev1995quantum}, with practical algorithms \cite[e.g.]{parker2020quantum,egger2020warmstarting} even for noisy quantum devices.
Indeed, one can show \cite{Wocjan2006} that any algorithm for a quantum computer with an exponential speed-up is reducible to an eigensolver.

There are two options for linearising the generalized eigenvalue problem, broadly speaking. Either we utilize the {\em companion
	linearization} 
	\cite{gohberg1982matrix,Mackey2006,higham2006conditioning}
	to transform the polynomial eigenvalue problem \eqref{eq:PEP} into a linear pencil in a higher dimension, or we utilize the congruence transformations \cite{Lucas2004}. Either way, we express the generalized eigenvalues in the generalized problem as the standard eigenvalues of a larger matrix. 

\subsection{Companion Linearization}

Let us consider the polynomial eigenvalue problem \eqref{eq:PEP}. 
Starting from the generalized eigenvalue problem $C_0 - \lambda C_1$, where the companion matrices \cite[Chapter 4]{gohberg1982matrix} are:
{ \scriptsize
  \[
	C_0=
	\left[
	\begin{array}{cccc}
	B_d & 0 & \cdots & 0 \\
	0 & I_m &  \ddots & \vdots \\
	\vdots & \ddots & \ddots & 0\\
	0 & \cdots & 0 & I_m
	\end{array}
	\right]
	\text{, }
	C_1 =
	\left[
	\begin{array}{cccc}
	B_{d-1} & B_{d-2} & \cdots & B_0 \\
	- I_m & 0 &  \cdots & 0 \\
	\vdots & \ddots & \ddots & \vdots\\
	0 & \cdots & - I_m & 0
	\end{array}
	\right] ,
	\]
}
\noindent
where $I_{m}$ denotes the $m \times m$ identity matrix.
we obtain the 
usual linear eigenvalue problem
$(\lambda I_d - C_2) z = 0$, where
\[
C_2 =
\left[
\begin{array}{cccc}
B_{d-1} B_d^{-1} & B_{d-2} B_d^{-1} & \cdots & B_0 B_d^{-1} \\
-I_m & 0 &  \cdots & 0 \\
\vdots & \ddots & \ddots & \vdots\\
0 & \cdots & -I_m & 0
\end{array}
\right] .
\]
The eigenvectors are roots of the characteristic polynomial of
$C_2$.

\begin{table}[!t]
    \caption{ A short history of options for translating generalized eigenvalue problems to eigenvalue problems.}\label{tab:GEP} 
\centering
    \begin{tabularx}{\columnwidth}{@{\extracolsep{\fill}}lllX}
    \hline
    Ref. & Year & Approach / Algorithm \\ \hline
\cite{gohberg1982matrix,berhanu2005polynomial} & folklore \; & Companion linearization  \\
\cite{parlett1971analysis} & 1971 & Three RRD (SPEC / SPEC / SVD) \\
\cite{fix1972algorithm} & 1972 & Three RRD (SPEC / SPEC / QR) \\
\cite{bunse1984algorithm} & 1984 & MDR \\ 
\cite{cao1987deflation} & 1987 & Three RRD (all orthogonal)  \\
\cite{demmel1993generalized} & 1993 & Generalized Upper Triangular (GUPTRI) \\ 
\cite{Lucas2004} & 2004 & Orthogonal RRD (SPEC / SPEC / SVD)  \\
\cite{Lucas2004} & 2004 & Non-orthogonal (Cholesky / LDL$^{T}$ / QRP) \\
    \hline
    \end{tabularx}
\end{table}
This approach is ready to be used on noisy quantum devices, in the sense that it does not require the implementation of any numerical linear algebra on the  quantum device, other than the eigensolver, and moreover, in that it is very robust to errors in the quantum eigensolver.

\subsection{Congruence Transformations}

An alternative approach is known as the congruence transformations. This stems from the work of Lucas \cite[Chapter 3]{Lucas2004}, which generalizes earlier work of Fix and Heiberger \cite{fix1972algorithm}, Parlett \cite{parlett1971analysis}, and Cao \cite{cao1987deflation}. 
We refer to Section 3.4.5 of \cite{Lucas2004} for the discussion of the computational complexity. 
\footnote{Essentially, this depends on the
rank-revealing decomposition.
Spectral decomposition required $4n^{3}$ flops, while Cholesky or LDL$^{T}$ require $n^3/3$ flops.}
While this work, summarized in Table \ref{tab:GEP}, is fundamental in (multi)linear algebra, it is surprisingly little known.
Having said that, it may be less suited to noisy quantum devices, in the sense that the quantum eigensolver gets compounded up to three times within the quantum boundary oracle. 

\section{Experimental Results}
\label{sec:results}

We have implemented the random-walk variant of the cutting-plane method specialized to SDPs in Python, with a view of inclusion of the code in Qiskit \cite{Qiskit}.
The pseudo code of the algorithms is presented in Appendix \ref{sec:impl}, while numerical constants and other details of the implementation are discussed in Appendix~\ref{sec:details}.

We have tested our implementation on SDPLIB \cite{borchers1999sdplib}, a well-known benchmark.
Table \ref{tab:sdplib1} presents an overview of the solution quality obtained on a subset of the instances.
We should like to stress that the SDPLIB has been designed to test the scalability of classical interior-point methods, and 
while it may provide the ultimate test of scalability of quantum algorithms for semidefinite programming, 
none of the quantum algorithms surveyed in Table \ref{tab:cvxopt} has been tested on any instances from SDPLIB, yet.
Likewise, while there has been much effort focussed on implementations \cite{mittelmann2003independent,chalkis2020efficient} of SDP solvers, we believe these to be the first reported results of a cutting plane method on the SDPLIB.

As can be seen in Table \ref{tab:sdplib1}, our method is much slower than classical interior-point methods. 
For example, on the instance
hinf1, which has been originally developed  by P. Gahinet within control-theoretic applications,
using a $14 \times 14$ PSD matrix and 13 inequalities, our method converges to 2 significant digits in the objective function within 127 seconds. 
In contrast, 
a commonly-used classical solver SCS 2.1.4 \cite{scs,scs16} solves the hinf1 to 3 significant digits in the objective function within 5.68 seconds on the same hardware; many interior-point methods \cite{mittelmann2003independent} are much faster still. 
As we detail in Table~\ref{tab:sdplib2} in Appendix~\ref{app:res}, 
on many other instances, our method terminates after 24 hours without obtaining a solution matching 1 significant digit in the value of the objective function.
Despite the appealing iteration-complexity results for cutting-plane methods, cf. Table \ref{tab:cutting_plane_method}, their practical utility remains limited, when executed classically. 


In terms of a potential quantum speed-up, much depends on the speed-up of the eigensolver, as discussed in Section~\ref{sec:main}.
For instance, on qap6, which features a $37 \times 37$ PSD matrix variable, approximately 8 hours and 24 minutes are spent in the eigensolver, classically. A square root of the run-time, which could illustrate a quadratic speed-up in a realistically implementable quantum eigensolver \cite[e.g.]{somma2013spectral,parker2020quantum}, would reduce this to less than 174 seconds. 
A logarithmic reduction of the run-time, which could illustrate the impact of an exponential quantum speed-up \footnote{An exponential quantum speed-up claimed by \cite{lloyd2014quantum} only under very particular circumstances, incl. low-rank matrices and strong assumptions on the initialization, has since been disputed \cite{Tang2018,Tang2020,Tang2021,chepurko2020quantum}. We do \emph{not} claim an exponential quantum speed-up is available.}, would reduce this to less than 5 seconds. 
A commonly used classical solver SCS 2.1.4 \cite{scs,scs16} solves the qap6 within 1.49 seconds on the same classical hardware. 

\begin{figure}
\centering
\includegraphics[width=0.45\textwidth]{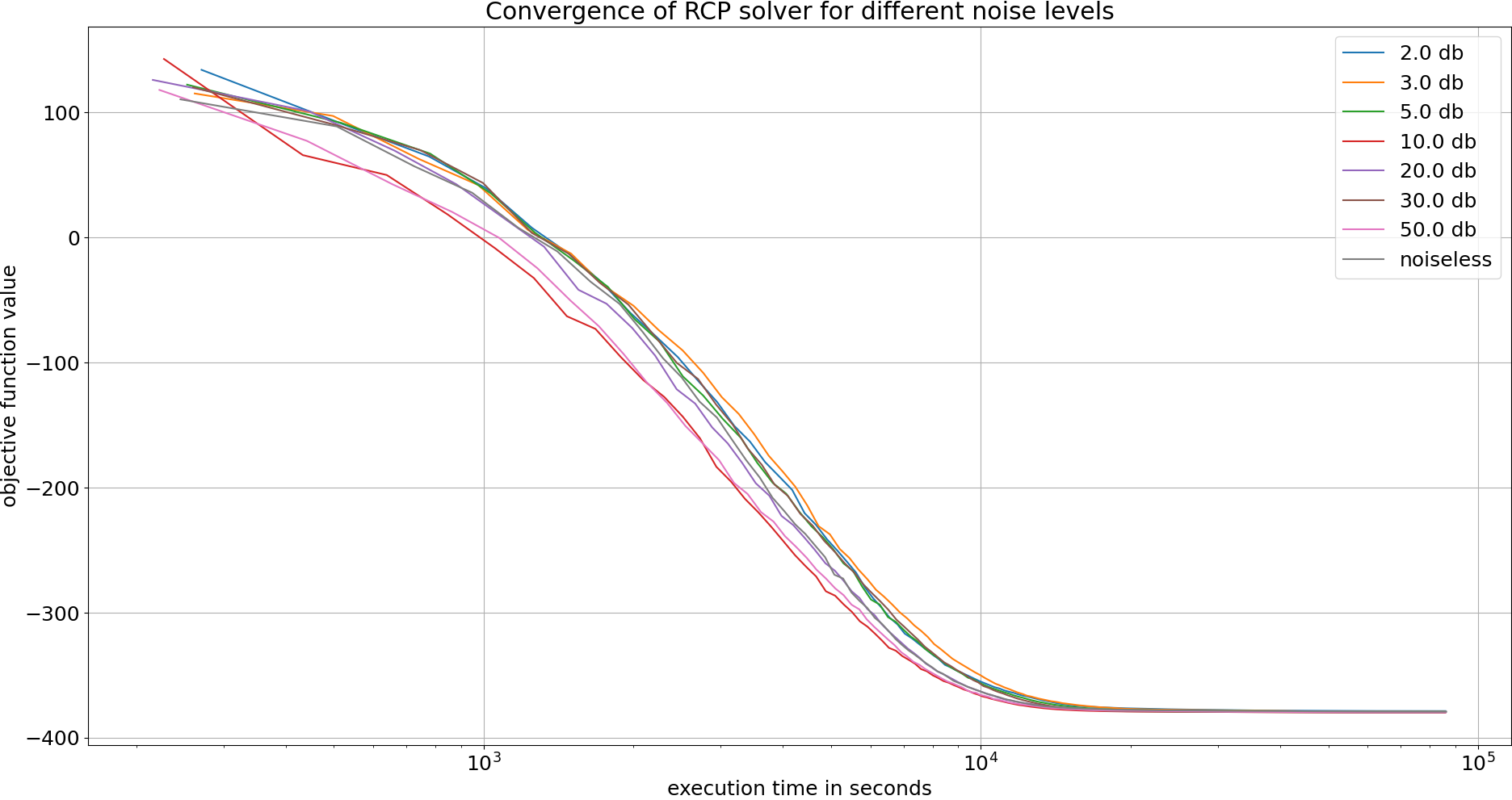}
\caption{The effect of errors in the eigensolver on the instance \texttt{qap6}, as applied with different signal to noise ratios (SNR). The curves demonstrate the evolution of objective function value in time.}
\label{fig:noise}
\end{figure}

Even the exponential speed-up in the eigensolver would hence yield a speed-up of the overall cutting-plane method only for instances (much) larger than a $37 \times 37$ PSD matrix, whilst our current ability to realize any quantum speed-up \cite{lloyd2014quantum,parker2020quantum} whatsoever in an eigensolver for an $37 \times 37$ matrix is lacking. Indeed, even the state preparation for a $37 \times 37$ matrix is presently out of reach.
Still, should the state preparation for larger matrices prove feasible, the overall speed-up may be of interest.

On a more positive note, in terms of the robustness to errors in the boundary oracle, which could be implemented with a quantum eigensolver as discussed in Section~\ref{sec:main}, the random-walk variant of the cutting-plane method may be hard to improve upon. As we illustrate in Figure~\ref{fig:noise}, even multiplicative noise in the eigenvalue computation corresponding to  the signal-to-noise ratio of approximately 2 dB does not change the performance of the algorithm on qap6, substantially.
(See Appendix \ref{sec:details} for the details of the noise model.)
This also has an intuitive interpretation, when one recalls that we use the boundary oracle to estimate the line segment along a sampled random direction that lies within the feasible set. We do not, however, use the estimated end points of the line segment \emph{per se}: we only sample from the line segment. Unless the error in the eigensolve leads to sampling from beyond the line segment, outside of the feasible set, the error has no discernible impact on the performance and does not propagate further. 
The fact that we can accommodate a substantial amount of noise in the quantum eigensolver 
could be seen as a basis of an approach suitable for noisy quantum devices. 


\begin{table*}[t!]
\caption{An overview of the solution quality obtained within a 24-hour time limit on a subset of smaller problems from SDPLIB, a well-known benchmark:
Instance name, its variant, number of constraints and size of the matrix  variable, reference objective function value in SDPLIB, 
 objective function value obtained by a classical solver (SCS, \cite{scs,scs16}),
and the objective-function value at the termination of RCP (in parentheses, if the run terminated with a time out).
Reasons for termination are detailed in Appendix~\ref{tab:sdplib2}. 
}
\label{tab:sdplib1}
\begin{tabularx}{\textwidth}{@{\extracolsep{\fill}}llrrrrr}
    \hline
 Instance & Type & $m$ & $n \times n$ &  Ref. Obj. &  SCS Obj. &  RCP Obj.  \\
    \hline
hinf1  &  primal  &  13 & $14 \times 14$  &   2.032600  &   2.036422  &   2.253482  \\
hinf10  &  primal  &  21 & $18 \times 18$  &   108.711800  &   107.808783  &   151.069581  \\
mcp100  &  primal  &   100 & $100 \times 100$  &   226.157350  &   226.151147  &   318.123641  \\
mcp124-1  &  primal  &   124 & $124 \times 124$  &   141.990480  &   141.988895  &   (250.149694)  \\
mcp124-2  &  primal  &   124 & $124 \times 124$  &   269.880170  &   269.885616  &   (383.003507)  \\
mcp124-3  &  primal  &   124 & $124 \times 124$  &   467.750110  &   467.755324  &   (614.279456)  \\
mcp124-4  &  primal  &   124 & $124 \times 124$  &   864.411860  &   864.413665  &  (1120.640088)  \\
mcp250-2  &  primal  &   250 & $250 \times 250$  &   531.930080  &   531.928325  &   (758.006964)  \\
mcp250-3  &  primal  &   250 & $250 \times 250$  &   981.172570  &   981.185856  &  (1272.445273)  \\
mcp250-4  &  primal  &   250 & $250 \times 250$  &  1681.960100  &  1681.959153  &  (2155.741152) \\ 
truss1  &  primal  &   6 & $13 \times 13$  &  -8.999996  &  -8.999996  &  -7.101001  \\
truss2  &  primal  &  58 & $133 \times 133$  &  -123.380360  &  -123.376950  &   (-21.458773) \\
truss3  &  primal  &  27 & $31 \times 31$  &  -9.109996  &  -9.110169  &  -5.232127 \\ 
truss4  &  primal  &  12 & $19 \times 19$  &  -9.009996  &  -9.010000  &  -5.826993 \\
    \hline
\end{tabularx}
\end{table*}

\section{Conclusions}
\label{sec:conclusion}

We have demonstrated how to utilize  eigensolvers in solving semidefinite programs, which are perhaps the broadest widely used class of convex optimization problems. 
The resulting randomized cutting plane method has several non-trivial steps, with several design choices for each step, as documented in Tables~	\ref{tab:walks}--\ref{tab:GEP}. This may hence suggest something of a framework for the development of further related algorithms, by varying the design choices we made. 

\vskip 6mm 

\section*{Acknowledgements}
The authors acknowledge a substantial contribution by Cunlu Zhou during his internship at IBM Research in the summer of 2018, including the first implementation. Cunlu chose not to be listed as a co-author. 
Jakub acknowledges useful discussions with Joran Van Apeldoorn, Andr{\'a}s Gily{\'e}n, Sander Gribling,
Martin Mevissen, and Jiri Vala.
We also acknowledge that \cite{bharti2021nisq} has appeared at a similar time, with a different algorithm, also applicable to noisy quantum devices, but without any run-time bounds. 

IBM, the IBM logo, and ibm.com are trademarks of International Business Machines Corp., registered in many jurisdictions worldwide. Other product and service names might be trademarks of IBM or other companies. The current list of IBM trademarks is available at https://www.ibm.com/legal/copytrade.

Jakub Mare\v{c}ek's research has been supported by the OP VVV project CZ.02.1.01/0.0/0.0/16 019/0000765 Research Center for Informatics.

This work has also received funding from the Disruptive Technologies Innovation Fund (DTIF), by Enterprise Ireland, under project number DTIF2019-090.

\bibliography{refs.bib}

\begin{thebibliography}{106}%
\makeatletter
\providecommand \@ifxundefined [1]{%
 \@ifx{#1\undefined}
}%
\providecommand \@ifnum [1]{%
 \ifnum #1\expandafter \@firstoftwo
 \else \expandafter \@secondoftwo
 \fi
}%
\providecommand \@ifx [1]{%
 \ifx #1\expandafter \@firstoftwo
 \else \expandafter \@secondoftwo
 \fi
}%
\providecommand \natexlab [1]{#1}%
\providecommand \enquote  [1]{``#1''}%
\providecommand \bibnamefont  [1]{#1}%
\providecommand \bibfnamefont [1]{#1}%
\providecommand \citenamefont [1]{#1}%
\providecommand \href@noop [0]{\@secondoftwo}%
\providecommand \href [0]{\begingroup \@sanitize@url \@href}%
\providecommand \@href[1]{\@@startlink{#1}\@@href}%
\providecommand \@@href[1]{\endgroup#1\@@endlink}%
\providecommand \@sanitize@url [0]{\catcode `\\12\catcode `\$12\catcode
  `\&12\catcode `\#12\catcode `\^12\catcode `\_12\catcode `\%12\relax}%
\providecommand \@@startlink[1]{}%
\providecommand \@@endlink[0]{}%
\providecommand \url  [0]{\begingroup\@sanitize@url \@url }%
\providecommand \@url [1]{\endgroup\@href {#1}{\urlprefix }}%
\providecommand \urlprefix  [0]{URL }%
\providecommand \Eprint [0]{\href }%
\providecommand \doibase [0]{http://dx.doi.org/}%
\providecommand \selectlanguage [0]{\@gobble}%
\providecommand \bibinfo  [0]{\@secondoftwo}%
\providecommand \bibfield  [0]{\@secondoftwo}%
\providecommand \translation [1]{[#1]}%
\providecommand \BibitemOpen [0]{}%
\providecommand \bibitemStop [0]{}%
\providecommand \bibitemNoStop [0]{.\EOS\space}%
\providecommand \EOS [0]{\spacefactor3000\relax}%
\providecommand \BibitemShut  [1]{\csname bibitem#1\endcsname}%
\let\auto@bib@innerbib\@empty
\bibitem [{\citenamefont {Boyd}\ and\ \citenamefont
  {Vandenberghe}(2004)}]{boyd2004convex}%
  \BibitemOpen
  \bibfield  {author} {\bibinfo {author} {\bibfnamefont {Stephen~P}\
  \bibnamefont {Boyd}}\ and\ \bibinfo {author} {\bibfnamefont {Lieven}\
  \bibnamefont {Vandenberghe}},\ }\href@noop {} {\emph {\bibinfo {title}
  {Convex optimization}}}\ (\bibinfo  {publisher} {Cambridge university
  press},\ \bibinfo {year} {2004})\BibitemShut {NoStop}%
\bibitem [{Note1()}]{Note1}%
  \BibitemOpen
  \bibinfo {note} {Indeed, if $f(x)$ is convex, \begin {equation*} \begin
  {aligned} f(x) &\to \protect \qopname \relax m{min}\\ x &\in \protect
  \mathcal {X}, \end {aligned} \end {equation*} by introducing a slack variable
  $t$, we obtain an equivalent optimization problem of the form (\ref {eq:1.1})
  \begin {equation*} \begin {aligned} & \protect \qopname \relax m{min}t\\ x
  &\in \protect \mathcal {X},\protect \tmspace +\thinmuskip {.1667em} f(x) - t
  \leq 0 \end {aligned} \end {equation*}}\BibitemShut {NoStop}%
\bibitem [{\citenamefont {Anjos}\ and\ \citenamefont
  {Lasserre}(2011)}]{anjos2011handbook}%
  \BibitemOpen
  \bibfield  {author} {\bibinfo {author} {\bibfnamefont {Miguel~F}\
  \bibnamefont {Anjos}}\ and\ \bibinfo {author} {\bibfnamefont {Jean~B.}\
  \bibnamefont {Lasserre}},\ }\href {https://doi.org/10.1007/978-1-4614-0769-0}
  {\emph {\bibinfo {title} {Handbook on semidefinite, conic and polynomial
  optimization}}},\ Vol.\ \bibinfo {volume} {166}\ (\bibinfo  {publisher}
  {Springer Science \& Business Media},\ \bibinfo {year} {2011})\BibitemShut
  {NoStop}%
\bibitem [{Note2()}]{Note2}%
  \BibitemOpen
  \bibinfo {note} {We still assume there are two Euclidean balls $\protect
  \mathcal {B}_r$ and $\protect \mathcal {B}_R$ inscribed and outscribed to the
  feasible set of \ref {eq:SDP}, which is known as the spectrahedron. Radius
  $R$ of the ball $\protect \mathcal {B}_R$ above can also be seen as an upper
  bound on the trace of an optimal primal solution of an SDP. Notice that in
  the case of a general SDP, parameters $r$ and $R$ are not constants
  independent of dimension, but do grow with the dimension.}\BibitemShut
  {Stop}%
\bibitem [{\citenamefont {Majumdar}\ \emph {et~al.}(2020)\citenamefont
  {Majumdar}, \citenamefont {Hall},\ and\ \citenamefont
  {Ahmadi}}]{majumdar2020recent}%
  \BibitemOpen
  \bibfield  {author} {\bibinfo {author} {\bibfnamefont {Anirudha}\
  \bibnamefont {Majumdar}}, \bibinfo {author} {\bibfnamefont {Georgina}\
  \bibnamefont {Hall}}, \ and\ \bibinfo {author} {\bibfnamefont {Amir~Ali}\
  \bibnamefont {Ahmadi}},\ }\bibfield  {title} {\enquote {\bibinfo {title}
  {Recent scalability improvements for semidefinite programming with
  applications in machine learning, control, and robotics},}\ }\href
  {https://doi.org/10.1146/annurev-control-091819-074326} {\bibfield  {journal}
  {\bibinfo  {journal} {Annual Review of Control, Robotics, and Autonomous
  Systems}\ }\textbf {\bibinfo {volume} {3}},\ \bibinfo {pages} {331--360}
  (\bibinfo {year} {2020})}\BibitemShut {NoStop}%
\bibitem [{\citenamefont {Khot}\ \emph {et~al.}(2007)\citenamefont {Khot},
  \citenamefont {Kindler}, \citenamefont {Mossel},\ and\ \citenamefont
  {O'Donnell}}]{khot2007optimal}%
  \BibitemOpen
  \bibfield  {author} {\bibinfo {author} {\bibfnamefont {Subhash}\ \bibnamefont
  {Khot}}, \bibinfo {author} {\bibfnamefont {Guy}\ \bibnamefont {Kindler}},
  \bibinfo {author} {\bibfnamefont {Elchanan}\ \bibnamefont {Mossel}}, \ and\
  \bibinfo {author} {\bibfnamefont {Ryan}\ \bibnamefont {O'Donnell}},\
  }\bibfield  {title} {\enquote {\bibinfo {title} {Optimal inapproximability
  results for {MAX-CUT} and other 2-variable {CSP}s?}}\ }\href
  {https://doi.org/10.1137/S0097539705447372} {\bibfield  {journal} {\bibinfo
  {journal} {SIAM Journal on Computing}\ }\textbf {\bibinfo {volume} {37}},\
  \bibinfo {pages} {319--357} (\bibinfo {year} {2007})}\BibitemShut {NoStop}%
\bibitem [{\citenamefont {Khot}(2010)}]{KhotSurvey}%
  \BibitemOpen
  \bibfield  {author} {\bibinfo {author} {\bibfnamefont {S.}~\bibnamefont
  {Khot}},\ }\bibfield  {title} {\enquote {\bibinfo {title} {On the unique
  games conjecture (invited survey)},}\ }in\ \href {\doibase
  10.1109/CCC.2010.19} {\emph {\bibinfo {booktitle} {2012 IEEE 27th Conference
  on Computational Complexity}}}\ (\bibinfo  {publisher} {IEEE Computer
  Society},\ \bibinfo {address} {Los Alamitos, CA, USA},\ \bibinfo {year}
  {2010})\ pp.\ \bibinfo {pages} {99--121}\BibitemShut {NoStop}%
\bibitem [{\citenamefont {Khot}\ and\ \citenamefont
  {Vishnoi}(2015)}]{khot2015unique}%
  \BibitemOpen
  \bibfield  {author} {\bibinfo {author} {\bibfnamefont {Subhash~A}\
  \bibnamefont {Khot}}\ and\ \bibinfo {author} {\bibfnamefont {Nisheeth~K}\
  \bibnamefont {Vishnoi}},\ }\bibfield  {title} {\enquote {\bibinfo {title}
  {The unique games conjecture, integrality gap for cut problems and
  embeddability of negative-type metrics into l1},}\ }\href
  {http://dx.doi.org/10.1145/2629614} {\bibfield  {journal} {\bibinfo
  {journal} {Journal of the ACM (JACM)}\ }\textbf {\bibinfo {volume} {62}},\
  \bibinfo {pages} {1--39} (\bibinfo {year} {2015})}\BibitemShut {NoStop}%
\bibitem [{\citenamefont {Raghavan}\ and\ \citenamefont
  {Tompson}(1987)}]{Raghavan1987}%
  \BibitemOpen
  \bibfield  {author} {\bibinfo {author} {\bibfnamefont {Prabhakar}\
  \bibnamefont {Raghavan}}\ and\ \bibinfo {author} {\bibfnamefont {Clark~D.}\
  \bibnamefont {Tompson}},\ }\bibfield  {title} {\enquote {\bibinfo {title}
  {Randomized rounding: A technique for provably good algorithms and
  algorithmic proofs},}\ }\href {\doibase 10.1007/BF02579324} {\bibfield
  {journal} {\bibinfo  {journal} {Combinatorica}\ }\textbf {\bibinfo {volume}
  {7}},\ \bibinfo {pages} {365--374} (\bibinfo {year} {1987})}\BibitemShut
  {NoStop}%
\bibitem [{\citenamefont {Jiang}\ \emph
  {et~al.}(2020{\natexlab{a}})\citenamefont {Jiang}, \citenamefont {Kathuria},
  \citenamefont {Lee}, \citenamefont {Padmanabhan},\ and\ \citenamefont
  {Song}}]{jiang2020faster}%
  \BibitemOpen
  \bibfield  {author} {\bibinfo {author} {\bibfnamefont {Haotian}\ \bibnamefont
  {Jiang}}, \bibinfo {author} {\bibfnamefont {Tarun}\ \bibnamefont {Kathuria}},
  \bibinfo {author} {\bibfnamefont {Yin~Tat}\ \bibnamefont {Lee}}, \bibinfo
  {author} {\bibfnamefont {Swati}\ \bibnamefont {Padmanabhan}}, \ and\ \bibinfo
  {author} {\bibfnamefont {Zhao}\ \bibnamefont {Song}},\ }\bibfield  {title}
  {\enquote {\bibinfo {title} {A faster interior point method for semidefinite
  programming},}\ }in\ \href@noop {} {\emph {\bibinfo {booktitle} {2020 IEEE
  61st Annual Symposium on Foundations of Computer Science (FOCS)}}}\ (\bibinfo
  {organization} {IEEE},\ \bibinfo {year} {2020})\ pp.\ \bibinfo {pages}
  {910--918}\BibitemShut {NoStop}%
\bibitem [{\citenamefont {Mosek}(2020)}]{mosek2020}%
  \BibitemOpen
  \bibfield  {author} {\bibinfo {author} {\bibfnamefont {APS}\ \bibnamefont
  {Mosek}},\ }\href@noop {} {\enquote {\bibinfo {title} {The mosek optimization
  software},}\ } (\bibinfo {year} {2020}),\ \bibinfo {note} {online at
  http://www.mosek.com}\BibitemShut {NoStop}%
\bibitem [{\citenamefont {Blum}\ \emph {et~al.}(2012)\citenamefont {Blum},
  \citenamefont {Cucker}, \citenamefont {Shub},\ and\ \citenamefont
  {Smale}}]{blum2012complexity}%
  \BibitemOpen
  \bibfield  {author} {\bibinfo {author} {\bibfnamefont {Lenore}\ \bibnamefont
  {Blum}}, \bibinfo {author} {\bibfnamefont {Felipe}\ \bibnamefont {Cucker}},
  \bibinfo {author} {\bibfnamefont {Michael}\ \bibnamefont {Shub}}, \ and\
  \bibinfo {author} {\bibfnamefont {Steve}\ \bibnamefont {Smale}},\ }\href
  {https://doi.org/10.1007/978-1-4612-0701-6} {\emph {\bibinfo {title}
  {Complexity and real computation}}}\ (\bibinfo  {publisher} {Springer Science
  \& Business Media},\ \bibinfo {year} {2012})\BibitemShut {NoStop}%
\bibitem [{\citenamefont {Yurtsever}\ \emph {et~al.}(2019)\citenamefont
  {Yurtsever}, \citenamefont {Tropp}, \citenamefont {Fercoq}, \citenamefont
  {Udell},\ and\ \citenamefont {Cevher}}]{yurtsever2019scalable}%
  \BibitemOpen
  \bibfield  {author} {\bibinfo {author} {\bibfnamefont {Alp}\ \bibnamefont
  {Yurtsever}}, \bibinfo {author} {\bibfnamefont {Joel~A}\ \bibnamefont
  {Tropp}}, \bibinfo {author} {\bibfnamefont {Olivier}\ \bibnamefont {Fercoq}},
  \bibinfo {author} {\bibfnamefont {Madeleine}\ \bibnamefont {Udell}}, \ and\
  \bibinfo {author} {\bibfnamefont {Volkan}\ \bibnamefont {Cevher}},\
  }\bibfield  {title} {\enquote {\bibinfo {title} {Scalable semidefinite
  programming},}\ }\href {https://arxiv.org/abs/1912.02949} {\bibfield
  {journal} {\bibinfo  {journal} {arXiv:1912.02949}\ } (\bibinfo {year}
  {2019})}\BibitemShut {NoStop}%
\bibitem [{\citenamefont {Ganzhorn}\ \emph {et~al.}(2019)\citenamefont
  {Ganzhorn}, \citenamefont {Egger}, \citenamefont {Barkoutsos}, \citenamefont
  {Ollitrault}, \citenamefont {Salis}, \citenamefont {Moll}, \citenamefont
  {Fuhrer}, \citenamefont {Mueller}, \citenamefont {Woerner}, \citenamefont
  {Tavernelli},\ and\ \citenamefont {Filipp}}]{Ganzhorn2019}%
  \BibitemOpen
  \bibfield  {author} {\bibinfo {author} {\bibfnamefont {Marc}\ \bibnamefont
  {Ganzhorn}}, \bibinfo {author} {\bibfnamefont {Daniel~J.}\ \bibnamefont
  {Egger}}, \bibinfo {author} {\bibfnamefont {Panagiotis~Kl.}\ \bibnamefont
  {Barkoutsos}}, \bibinfo {author} {\bibfnamefont {Pauline}\ \bibnamefont
  {Ollitrault}}, \bibinfo {author} {\bibfnamefont {Gian}\ \bibnamefont
  {Salis}}, \bibinfo {author} {\bibfnamefont {Nikolaj}\ \bibnamefont {Moll}},
  \bibinfo {author} {\bibfnamefont {Andreas}\ \bibnamefont {Fuhrer}}, \bibinfo
  {author} {\bibfnamefont {Peter}\ \bibnamefont {Mueller}}, \bibinfo {author}
  {\bibfnamefont {Stefan}\ \bibnamefont {Woerner}}, \bibinfo {author}
  {\bibfnamefont {Ivano}\ \bibnamefont {Tavernelli}}, \ and\ \bibinfo {author}
  {\bibfnamefont {Stefan}\ \bibnamefont {Filipp}},\ }\bibfield  {title}
  {\enquote {\bibinfo {title} {Gate-efficient simulation of molecular
  eigenstates on a quantum computer},}\ }\href {\doibase
  10.1103/PhysRevApplied.11.044092} {\bibfield  {journal} {\bibinfo  {journal}
  {Phys. Rev. Applied}\ }\textbf {\bibinfo {volume} {11}},\ \bibinfo {pages}
  {044092} (\bibinfo {year} {2019})}\BibitemShut {NoStop}%
\bibitem [{\citenamefont {Havlicek}\ \emph {et~al.}(2019)\citenamefont
  {Havlicek}, \citenamefont {Corcoles}, \citenamefont {Temme}, \citenamefont
  {Harrow}, \citenamefont {Kandala}, \citenamefont {Chow},\ and\ \citenamefont
  {Gambetta}}]{Havlicek2019}%
  \BibitemOpen
  \bibfield  {author} {\bibinfo {author} {\bibfnamefont {Vojtech}\ \bibnamefont
  {Havlicek}}, \bibinfo {author} {\bibfnamefont {Antonio~D.}\ \bibnamefont
  {Corcoles}}, \bibinfo {author} {\bibfnamefont {Kristan}\ \bibnamefont
  {Temme}}, \bibinfo {author} {\bibfnamefont {Aram~W.}\ \bibnamefont {Harrow}},
  \bibinfo {author} {\bibfnamefont {Abhinav}\ \bibnamefont {Kandala}}, \bibinfo
  {author} {\bibfnamefont {Jerry~M.}\ \bibnamefont {Chow}}, \ and\ \bibinfo
  {author} {\bibfnamefont {Jay~M.}\ \bibnamefont {Gambetta}},\ }\bibfield
  {title} {\enquote {\bibinfo {title} {Supervised learning with
  quantum-enhanced feature spaces},}\ }\href {\doibase DO -
  10.1038/s41586-019-0980-2} {\bibfield  {journal} {\bibinfo  {journal}
  {Nature}\ }\textbf {\bibinfo {volume} {567}},\ \bibinfo {pages} {209 -- 212}
  (\bibinfo {year} {2019})}\BibitemShut {NoStop}%
\bibitem [{\citenamefont {Egger}\ \emph {et~al.}(2020)\citenamefont {Egger},
  \citenamefont {Gambella}, \citenamefont {Marecek}, \citenamefont {McFaddin},
  \citenamefont {Mevissen}, \citenamefont {Raymond}, \citenamefont {Simonetto},
  \citenamefont {Woerner},\ and\ \citenamefont {Yndurain}}]{egger2020quantum}%
  \BibitemOpen
  \bibfield  {author} {\bibinfo {author} {\bibfnamefont {Daniel~J}\
  \bibnamefont {Egger}}, \bibinfo {author} {\bibfnamefont {Claudio}\
  \bibnamefont {Gambella}}, \bibinfo {author} {\bibfnamefont {Jakub}\
  \bibnamefont {Marecek}}, \bibinfo {author} {\bibfnamefont {Scott}\
  \bibnamefont {McFaddin}}, \bibinfo {author} {\bibfnamefont {Martin}\
  \bibnamefont {Mevissen}}, \bibinfo {author} {\bibfnamefont {Rudy}\
  \bibnamefont {Raymond}}, \bibinfo {author} {\bibfnamefont {Andrea}\
  \bibnamefont {Simonetto}}, \bibinfo {author} {\bibfnamefont {Stefan}\
  \bibnamefont {Woerner}}, \ and\ \bibinfo {author} {\bibfnamefont {Elena}\
  \bibnamefont {Yndurain}},\ }\bibfield  {title} {\enquote {\bibinfo {title}
  {Quantum computing for finance: state of the art and future prospects},}\
  }\href@noop {} {\bibfield  {journal} {\bibinfo  {journal} {IEEE Transactions
  on Quantum Engineering}\ } (\bibinfo {year} {2020})}\BibitemShut {NoStop}%
\bibitem [{\citenamefont {Garg}\ \emph {et~al.}(2021)\citenamefont {Garg},
  \citenamefont {Kothari}, \citenamefont {Netrapalli},\ and\ \citenamefont
  {Sherif}}]{garg2021}%
  \BibitemOpen
  \bibfield  {author} {\bibinfo {author} {\bibfnamefont {Ankit}\ \bibnamefont
  {Garg}}, \bibinfo {author} {\bibfnamefont {Robin}\ \bibnamefont {Kothari}},
  \bibinfo {author} {\bibfnamefont {Praneeth}\ \bibnamefont {Netrapalli}}, \
  and\ \bibinfo {author} {\bibfnamefont {Suhail}\ \bibnamefont {Sherif}},\
  }\bibfield  {title} {\enquote {\bibinfo {title} {{No Quantum Speedup over
  Gradient Descent for Non-Smooth Convex Optimization}},}\ }in\ \href {\doibase
  10.4230/LIPIcs.ITCS.2021.53} {\emph {\bibinfo {booktitle} {12th Innovations
  in Theoretical Computer Science Conference (ITCS 2021)}}},\ \bibinfo {series}
  {Leibniz International Proceedings in Informatics (LIPIcs)}, Vol.\ \bibinfo
  {volume} {185},\ \bibinfo {editor} {edited by\ \bibinfo {editor}
  {\bibfnamefont {James~R.}\ \bibnamefont {Lee}}}\ (\bibinfo  {publisher}
  {Schloss Dagstuhl--Leibniz-Zentrum f{\"u}r Informatik},\ \bibinfo {address}
  {Dagstuhl, Germany},\ \bibinfo {year} {2021})\ pp.\ \bibinfo {pages}
  {53:1--53:20}\BibitemShut {NoStop}%
\bibitem [{\citenamefont {{Brandao}}\ and\ \citenamefont
  {{Svore}}(2017)}]{8104077}%
  \BibitemOpen
  \bibfield  {author} {\bibinfo {author} {\bibfnamefont {F.~G. S.~L.}\
  \bibnamefont {{Brandao}}}\ and\ \bibinfo {author} {\bibfnamefont {K.~M.}\
  \bibnamefont {{Svore}}},\ }\bibfield  {title} {\enquote {\bibinfo {title}
  {Quantum speed-ups for solving semidefinite programs},}\ }in\ \href {\doibase
  10.1109/FOCS.2017.45} {\emph {\bibinfo {booktitle} {2017 IEEE 58th Annual
  Symposium on Foundations of Computer Science (FOCS)}}}\ (\bibinfo {year}
  {2017})\ pp.\ \bibinfo {pages} {415--426}\BibitemShut {NoStop}%
\bibitem [{\citenamefont {{Van Apeldoorn}}\ \emph {et~al.}(2017)\citenamefont
  {{Van Apeldoorn}}, \citenamefont {{Gilyén}}, \citenamefont {{Gribling}},\
  and\ \citenamefont {{de Wolf}}}]{8104076}%
  \BibitemOpen
  \bibfield  {author} {\bibinfo {author} {\bibfnamefont {J.}~\bibnamefont {{Van
  Apeldoorn}}}, \bibinfo {author} {\bibfnamefont {A.}~\bibnamefont
  {{Gilyén}}}, \bibinfo {author} {\bibfnamefont {S.}~\bibnamefont
  {{Gribling}}}, \ and\ \bibinfo {author} {\bibfnamefont {R.}~\bibnamefont {{de
  Wolf}}},\ }\bibfield  {title} {\enquote {\bibinfo {title} {Quantum
  sdp-solvers: Better upper and lower bounds},}\ }in\ \href {\doibase
  10.1109/FOCS.2017.44} {\emph {\bibinfo {booktitle} {2017 IEEE 58th Annual
  Symposium on Foundations of Computer Science (FOCS)}}}\ (\bibinfo {year}
  {2017})\ pp.\ \bibinfo {pages} {403--414}\BibitemShut {NoStop}%
\bibitem [{\citenamefont {Brand{\~a}o}\ \emph {et~al.}(2019)\citenamefont
  {Brand{\~a}o}, \citenamefont {Kalev}, \citenamefont {Li}, \citenamefont
  {Lin}, \citenamefont {Svore},\ and\ \citenamefont {Wu}}]{brandao2019quantum}%
  \BibitemOpen
  \bibfield  {author} {\bibinfo {author} {\bibfnamefont {Fernando~GSL}\
  \bibnamefont {Brand{\~a}o}}, \bibinfo {author} {\bibfnamefont {Amir}\
  \bibnamefont {Kalev}}, \bibinfo {author} {\bibfnamefont {Tongyang}\
  \bibnamefont {Li}}, \bibinfo {author} {\bibfnamefont {Cedric Yen-Yu}\
  \bibnamefont {Lin}}, \bibinfo {author} {\bibfnamefont {Krysta~M}\
  \bibnamefont {Svore}}, \ and\ \bibinfo {author} {\bibfnamefont {Xiaodi}\
  \bibnamefont {Wu}},\ }\bibfield  {title} {\enquote {\bibinfo {title} {Quantum
  sdp solvers: Large speed-ups, optimality, and applications to quantum
  learning},}\ }in\ \href@noop {} {\emph {\bibinfo {booktitle} {46th
  International Colloquium on Automata, Languages, and Programming (ICALP
  2019)}}}\ (\bibinfo {organization} {Schloss Dagstuhl-Leibniz-Zentrum fuer
  Informatik},\ \bibinfo {year} {2019})\BibitemShut {NoStop}%
\bibitem [{\citenamefont {van Apeldoorn}\ and\ \citenamefont
  {Gilyén}(2019)}]{van2019improvements}%
  \BibitemOpen
  \bibfield  {author} {\bibinfo {author} {\bibfnamefont {Joran}\ \bibnamefont
  {van Apeldoorn}}\ and\ \bibinfo {author} {\bibfnamefont {András}\
  \bibnamefont {Gilyén}},\ }\bibfield  {title} {\enquote {\bibinfo {title}
  {Improvements in quantum sdp-solving with applications},}\ }in\ \href
  {\doibase 10.4230/LIPIcs.ICALP.2019.99} {\emph {\bibinfo {booktitle}
  {Proceedings of 46th International Colloquium on Automata, Languages, and
  Programming (ICALP 2019}}}\ (\bibinfo {year} {2019})\BibitemShut {NoStop}%
\bibitem [{\citenamefont {Arora}\ and\ \citenamefont
  {Kale}(2007)}]{arora2007combinatorial}%
  \BibitemOpen
  \bibfield  {author} {\bibinfo {author} {\bibfnamefont {Sanjeev}\ \bibnamefont
  {Arora}}\ and\ \bibinfo {author} {\bibfnamefont {Satyen}\ \bibnamefont
  {Kale}},\ }\bibfield  {title} {\enquote {\bibinfo {title} {A combinatorial,
  primal-dual approach to semidefinite programs},}\ }in\ \href@noop {} {\emph
  {\bibinfo {booktitle} {Proceedings of the thirty-ninth annual ACM symposium
  on Theory of computing}}}\ (\bibinfo {year} {2007})\ pp.\ \bibinfo {pages}
  {227--236}\BibitemShut {NoStop}%
\bibitem [{\citenamefont {Hazan}(2008)}]{hazan2008sparse}%
  \BibitemOpen
  \bibfield  {author} {\bibinfo {author} {\bibfnamefont {Elad}\ \bibnamefont
  {Hazan}},\ }\bibfield  {title} {\enquote {\bibinfo {title} {Sparse
  approximate solutions to semidefinite programs},}\ }in\ \href@noop {} {\emph
  {\bibinfo {booktitle} {Latin American symposium on theoretical
  informatics}}}\ (\bibinfo {organization} {Springer},\ \bibinfo {year}
  {2008})\ pp.\ \bibinfo {pages} {306--316}\BibitemShut {NoStop}%
\bibitem [{Note3()}]{Note3}%
  \BibitemOpen
  \bibinfo {note} {We refer to Algorithm 6 in \cite {8104077} for a nice
  overview of the algorithm.}\BibitemShut {Stop}%
\bibitem [{\citenamefont {Kerenidis}\ and\ \citenamefont
  {Prakash}(2020)}]{kerenidis2018quantum}%
  \BibitemOpen
  \bibfield  {author} {\bibinfo {author} {\bibfnamefont {Iordanis}\
  \bibnamefont {Kerenidis}}\ and\ \bibinfo {author} {\bibfnamefont {Anupam}\
  \bibnamefont {Prakash}},\ }\bibfield  {title} {\enquote {\bibinfo {title} {A
  quantum interior point method for lps and sdps},}\ }\href {\doibase
  10.1145/3406306} {\bibfield  {journal} {\bibinfo  {journal} {ACM Transactions
  on Quantum Computing}\ }\textbf {\bibinfo {volume} {1}} (\bibinfo {year}
  {2020}),\ 10.1145/3406306}\BibitemShut {NoStop}%
\bibitem [{\citenamefont {Augustino}\ \emph {et~al.}(2021)\citenamefont
  {Augustino}, \citenamefont {Nannicini}, \citenamefont {Terlaky},\ and\
  \citenamefont {Zuluaga}}]{augustino2021inexact}%
  \BibitemOpen
  \bibfield  {author} {\bibinfo {author} {\bibfnamefont {Brandon}\ \bibnamefont
  {Augustino}}, \bibinfo {author} {\bibfnamefont {Giacomo}\ \bibnamefont
  {Nannicini}}, \bibinfo {author} {\bibfnamefont {Tam{\'a}s}\ \bibnamefont
  {Terlaky}}, \ and\ \bibinfo {author} {\bibfnamefont {Luis~F}\ \bibnamefont
  {Zuluaga}},\ }\bibfield  {title} {\enquote {\bibinfo {title} {An
  inexact-feasible quantum interior point method for semidefinite
  optimization},}\ }\href
  {https://engineering.lehigh.edu/sites/engineering.lehigh.edu/files/_DEPARTMENTS/ise/pdf/tech-papers/21/21T_005.pdf}
  {\  (\bibinfo {year} {2021})}\BibitemShut {NoStop}%
\bibitem [{\citenamefont {Wright}(1997)}]{wright1997primal}%
  \BibitemOpen
  \bibfield  {author} {\bibinfo {author} {\bibfnamefont {S.J.}\ \bibnamefont
  {Wright}},\ }\href {https://books.google.cz/books?id=2hANEPgs7CkC} {\emph
  {\bibinfo {title} {Primal-Dual Interior-Point Methods}}},\ Other Titles in
  Applied Mathematics\ (\bibinfo  {publisher} {Society for Industrial and
  Applied Mathematics},\ \bibinfo {year} {1997})\BibitemShut {NoStop}%
\bibitem [{Note4()}]{Note4}%
  \BibitemOpen
  \bibinfo {note} {We refer to Chapter 1 of \cite {wright1997primal} for an
  excellent introduction to primal-dual interior-point methods. Due to the
  reliance of reliance of interior-point methods on solving linear systems,
  \cite {kerenidis2018quantum,augustino2021inexact} ended up with a bound
  dependent on the condition number $\kappa $ of a linear system based on the
  Karush-Kuhn-Tucker (KKT) conditions. As is well known \cite [p.
  215]{wright1997primal}, this goes to infinity for all instances, by the
  design of the method, which may be not ideal in practice. Furthermore, there
  is the issue of the HHL algorithm \cite {harrow2009quantum} providing the
  solution of the linear system only as a quantum state, whereas the
  interior-point method \cite {kerenidis2018quantum,augustino2021inexact} needs
  a classical update. The HHL hence needs to be run many times and the quantum
  state measured many times, to estimate the classical update.}\BibitemShut
  {Stop}%
\bibitem [{\citenamefont {Gily{\'e}n}\ \emph {et~al.}(2019)\citenamefont
  {Gily{\'e}n}, \citenamefont {Arunachalam},\ and\ \citenamefont
  {Wiebe}}]{gilyen2019optimizing}%
  \BibitemOpen
  \bibfield  {author} {\bibinfo {author} {\bibfnamefont {Andr{\'a}s}\
  \bibnamefont {Gily{\'e}n}}, \bibinfo {author} {\bibfnamefont {Srinivasan}\
  \bibnamefont {Arunachalam}}, \ and\ \bibinfo {author} {\bibfnamefont
  {Nathan}\ \bibnamefont {Wiebe}},\ }\bibfield  {title} {\enquote {\bibinfo
  {title} {Optimizing quantum optimization algorithms via faster quantum
  gradient computation},}\ }in\ \href@noop {} {\emph {\bibinfo {booktitle}
  {Proceedings of the Thirtieth Annual ACM-SIAM Symposium on Discrete
  Algorithms}}}\ (\bibinfo {organization} {SIAM},\ \bibinfo {year} {2019})\
  pp.\ \bibinfo {pages} {1425--1444}\BibitemShut {NoStop}%
\bibitem [{\citenamefont {Van~Apeldoorn}\ \emph {et~al.}(2020)\citenamefont
  {Van~Apeldoorn}, \citenamefont {Gily{\'e}n}, \citenamefont {Gribling},\ and\
  \citenamefont {de~Wolf}}]{van2020quantum}%
  \BibitemOpen
  \bibfield  {author} {\bibinfo {author} {\bibfnamefont {Joran}\ \bibnamefont
  {Van~Apeldoorn}}, \bibinfo {author} {\bibfnamefont {Andr{\'a}s}\ \bibnamefont
  {Gily{\'e}n}}, \bibinfo {author} {\bibfnamefont {Sander}\ \bibnamefont
  {Gribling}}, \ and\ \bibinfo {author} {\bibfnamefont {Ronald}\ \bibnamefont
  {de~Wolf}},\ }\bibfield  {title} {\enquote {\bibinfo {title} {Quantum
  {SDP}-solvers: Better upper and lower bounds},}\ }\href
  {https://doi.org/10.22331/q-2020-02-14-230} {\bibfield  {journal} {\bibinfo
  {journal} {Quantum}\ }\textbf {\bibinfo {volume} {4}},\ \bibinfo {pages}
  {230} (\bibinfo {year} {2020})}\BibitemShut {NoStop}%
\bibitem [{\citenamefont {Chakrabarti}\ \emph {et~al.}(2020)\citenamefont
  {Chakrabarti}, \citenamefont {Childs}, \citenamefont {Li},\ and\
  \citenamefont {Wu}}]{chakrabarti2020quantum}%
  \BibitemOpen
  \bibfield  {author} {\bibinfo {author} {\bibfnamefont {Shouvanik}\
  \bibnamefont {Chakrabarti}}, \bibinfo {author} {\bibfnamefont {Andrew~M}\
  \bibnamefont {Childs}}, \bibinfo {author} {\bibfnamefont {Tongyang}\
  \bibnamefont {Li}}, \ and\ \bibinfo {author} {\bibfnamefont {Xiaodi}\
  \bibnamefont {Wu}},\ }\bibfield  {title} {\enquote {\bibinfo {title} {Quantum
  algorithms and lower bounds for convex optimization},}\ }\href@noop {}
  {\bibfield  {journal} {\bibinfo  {journal} {Quantum}\ }\textbf {\bibinfo
  {volume} {4}},\ \bibinfo {pages} {221} (\bibinfo {year} {2020})}\BibitemShut
  {NoStop}%
\bibitem [{Note5()}]{Note5}%
  \BibitemOpen
  \bibinfo {note} {As it has been shown in \cite [Appendix E]{van2020quantum},
  in the MWU algorithm, $\protect \frac {rR}{\epsilon }$ should be seen as an
  important parameter, as one can trade-off dependence on one of the three
  individual parameters for the dependence on the others.}\BibitemShut {Stop}%
\bibitem [{\citenamefont {Gr{\"o}tschel}\ \emph {et~al.}(2012)\citenamefont
  {Gr{\"o}tschel}, \citenamefont {Lov{\'a}sz},\ and\ \citenamefont
  {Schrijver}}]{grotschel2012geometric}%
  \BibitemOpen
  \bibfield  {author} {\bibinfo {author} {\bibfnamefont {Martin}\ \bibnamefont
  {Gr{\"o}tschel}}, \bibinfo {author} {\bibfnamefont {L{\'a}szl{\'o}}\
  \bibnamefont {Lov{\'a}sz}}, \ and\ \bibinfo {author} {\bibfnamefont
  {Alexander}\ \bibnamefont {Schrijver}},\ }\href@noop {} {\emph {\bibinfo
  {title} {Geometric algorithms and combinatorial optimization}}},\
  Vol.~\bibinfo {volume} {2}\ (\bibinfo  {publisher} {Springer Science \&
  Business Media},\ \bibinfo {year} {2012})\BibitemShut {NoStop}%
\bibitem [{\citenamefont {van Apeldoorn}(2020)}]{van2020phd}%
  \BibitemOpen
  \bibfield  {author} {\bibinfo {author} {\bibfnamefont {Joran}\ \bibnamefont
  {van Apeldoorn}},\ }\emph {\bibinfo {title} {A quantum view on convex
  optimization}},\ \href@noop {} {Ph.D. thesis},\ \bibinfo  {school}
  {University of Amsterdam Institute for Logic, Language and Computation
  (ILLC)} (\bibinfo {year} {2020})\BibitemShut {NoStop}%
\bibitem [{\citenamefont {van Apeldoorn}\ \emph {et~al.}(2020)\citenamefont
  {van Apeldoorn}, \citenamefont {Gilyén}, \citenamefont {Gribling},\ and\
  \citenamefont {de~Wolf}}]{van2020oracles}%
  \BibitemOpen
  \bibfield  {author} {\bibinfo {author} {\bibfnamefont {Joran}\ \bibnamefont
  {van Apeldoorn}}, \bibinfo {author} {\bibfnamefont {András}\ \bibnamefont
  {Gilyén}}, \bibinfo {author} {\bibfnamefont {Sander}\ \bibnamefont
  {Gribling}}, \ and\ \bibinfo {author} {\bibfnamefont {Ronald}\ \bibnamefont
  {de~Wolf}},\ }\bibfield  {title} {\enquote {\bibinfo {title} {Convex
  optimization using quantum oracles},}\ }\href {\doibase
  10.22331/q-2020-01-13-220} {\bibfield  {journal} {\bibinfo  {journal}
  {Quantum}\ }\textbf {\bibinfo {volume} {4}} (\bibinfo {year} {2020}),\
  10.22331/q-2020-01-13-220}\BibitemShut {NoStop}%
\bibitem [{\citenamefont {{Mohammad Hossein Mohammadi Siahroudi}}\ \emph
  {et~al.}(2021)\citenamefont {{Mohammad Hossein Mohammadi Siahroudi}},
  \citenamefont {Fakhimi},\ and\ \citenamefont
  {Terlaky}}]{mohammadisiahroudi2021efficient}%
  \BibitemOpen
  \bibfield  {author} {\bibinfo {author} {\bibnamefont {{Mohammad Hossein
  Mohammadi Siahroudi}}}, \bibinfo {author} {\bibfnamefont {Ramin}\
  \bibnamefont {Fakhimi}}, \ and\ \bibinfo {author} {\bibfnamefont {Tam{\'a}s}\
  \bibnamefont {Terlaky}},\ }\bibfield  {title} {\enquote {\bibinfo {title}
  {Efficient use of quantum linear system algorithms in interior point methods
  for linear optimization},}\ }\href
  {https://engineering.lehigh.edu/sites/engineering.lehigh.edu/files/_DEPARTMENTS/ise/pdf/tech-papers/21/21T_005.pdf}
  {\  (\bibinfo {year} {2021})}\BibitemShut {NoStop}%
\bibitem [{\citenamefont {Bharti}\ \emph {et~al.}(2021)\citenamefont {Bharti},
  \citenamefont {Haug}, \citenamefont {Vedral},\ and\ \citenamefont
  {Kwek}}]{bharti2021nisq}%
  \BibitemOpen
  \bibfield  {author} {\bibinfo {author} {\bibfnamefont {Kishor}\ \bibnamefont
  {Bharti}}, \bibinfo {author} {\bibfnamefont {Tobias}\ \bibnamefont {Haug}},
  \bibinfo {author} {\bibfnamefont {Vlatko}\ \bibnamefont {Vedral}}, \ and\
  \bibinfo {author} {\bibfnamefont {Leong-Chuan}\ \bibnamefont {Kwek}},\
  }\href@noop {} {\enquote {\bibinfo {title} {Nisq algorithm for semidefinite
  programming},}\ } (\bibinfo {year} {2021}),\ \Eprint
  {http://arxiv.org/abs/2106.03891} {arXiv:2106.03891 [quant-ph]} \BibitemShut
  {NoStop}%
\bibitem [{\citenamefont {Shor}(1977)}]{s77}%
  \BibitemOpen
  \bibfield  {author} {\bibinfo {author} {\bibfnamefont {Naum~Z}\ \bibnamefont
  {Shor}},\ }\bibfield  {title} {\enquote {\bibinfo {title} {Cut-off method
  with space extension in convex programming problems},}\ }\href@noop {}
  {\bibfield  {journal} {\bibinfo  {journal} {Cybernetics}\ }\textbf {\bibinfo
  {volume} {13}},\ \bibinfo {pages} {94--96} (\bibinfo {year}
  {1977})}\BibitemShut {NoStop}%
\bibitem [{\citenamefont {Yudin}\ and\ \citenamefont
  {Nemirovski}(1976)}]{yn76}%
  \BibitemOpen
  \bibfield  {author} {\bibinfo {author} {\bibfnamefont {David~B}\ \bibnamefont
  {Yudin}}\ and\ \bibinfo {author} {\bibfnamefont {Arkadii~S}\ \bibnamefont
  {Nemirovski}},\ }\bibfield  {title} {\enquote {\bibinfo {title} {Evaluation
  of the information complexity of mathematical programming problems},}\
  }\href@noop {} {\bibfield  {journal} {\bibinfo  {journal} {Ekonomika i
  Matematicheskie Metody}\ }\textbf {\bibinfo {volume} {13}},\ \bibinfo {pages}
  {128--142} (\bibinfo {year} {1976})}\BibitemShut {NoStop}%
\bibitem [{\citenamefont {Khachiyan}(1980)}]{k80}%
  \BibitemOpen
  \bibfield  {author} {\bibinfo {author} {\bibfnamefont {Leonid~G}\
  \bibnamefont {Khachiyan}},\ }\bibfield  {title} {\enquote {\bibinfo {title}
  {Polynomial algorithms in linear programming},}\ }\href@noop {} {\bibfield
  {journal} {\bibinfo  {journal} {USSR Computational Mathematics and
  Mathematical Physics}\ }\textbf {\bibinfo {volume} {20}},\ \bibinfo {pages}
  {53--72} (\bibinfo {year} {1980})}\BibitemShut {NoStop}%
\bibitem [{\citenamefont {Khachiyan}\ \emph {et~al.}(1988)\citenamefont
  {Khachiyan}, \citenamefont {Tarasov},\ and\ \citenamefont {Erlikh}}]{kte88}%
  \BibitemOpen
  \bibfield  {author} {\bibinfo {author} {\bibfnamefont {Leonid~G}\
  \bibnamefont {Khachiyan}}, \bibinfo {author} {\bibfnamefont
  {Sergei~Pavlovich}\ \bibnamefont {Tarasov}}, \ and\ \bibinfo {author}
  {\bibfnamefont {I.~I.}\ \bibnamefont {Erlikh}},\ }\bibfield  {title}
  {\enquote {\bibinfo {title} {The method of inscribed ellipsoids},}\ }in\
  \href@noop {} {\emph {\bibinfo {booktitle} {Soviet Math. Dokl}}},\
  Vol.~\bibinfo {volume} {37}\ (\bibinfo {year} {1988})\ pp.\ \bibinfo {pages}
  {226--230}\BibitemShut {NoStop}%
\bibitem [{\citenamefont {Nesterov}\ and\ \citenamefont
  {Nemirovskii}(1989)}]{nn89}%
  \BibitemOpen
  \bibfield  {author} {\bibinfo {author} {\bibfnamefont {Y.E.}\ \bibnamefont
  {Nesterov}}\ and\ \bibinfo {author} {\bibfnamefont {A.S.}\ \bibnamefont
  {Nemirovskii}},\ }\href@noop {} {\enquote {\bibinfo {title} {Self-concordant
  functions and polynomial-time methods in convex programming},}\ } (\bibinfo
  {year} {1989}),\ \bibinfo {note} {report, Central Economic and Mathematics
  Institute, USSR Acad. Sci}\BibitemShut {NoStop}%
\bibitem [{\citenamefont {Vaidya}(1989)}]{v89}%
  \BibitemOpen
  \bibfield  {author} {\bibinfo {author} {\bibfnamefont {Pravin~M}\
  \bibnamefont {Vaidya}},\ }\bibfield  {title} {\enquote {\bibinfo {title} {A
  new algorithm for minimizing convex functions over convex sets},}\ }in\
  \href@noop {} {\emph {\bibinfo {booktitle} {30th Annual Symposium on
  Foundations of Computer Science}}}\ (\bibinfo {organization} {IEEE Computer
  Society},\ \bibinfo {year} {1989})\ pp.\ \bibinfo {pages}
  {338--343}\BibitemShut {NoStop}%
\bibitem [{\citenamefont {Atkinson}\ and\ \citenamefont {Vaidya}(1995)}]{av95}%
  \BibitemOpen
  \bibfield  {author} {\bibinfo {author} {\bibfnamefont {David~S}\ \bibnamefont
  {Atkinson}}\ and\ \bibinfo {author} {\bibfnamefont {Pravin~M}\ \bibnamefont
  {Vaidya}},\ }\bibfield  {title} {\enquote {\bibinfo {title} {A cutting plane
  algorithm for convex programming that uses analytic centers},}\ }\href@noop
  {} {\bibfield  {journal} {\bibinfo  {journal} {Mathematical Programming}\
  }\textbf {\bibinfo {volume} {69}},\ \bibinfo {pages} {1--43} (\bibinfo {year}
  {1995})}\BibitemShut {NoStop}%
\bibitem [{\citenamefont {Bertsimas}\ and\ \citenamefont
  {Vempala}(2004)}]{bv02}%
  \BibitemOpen
  \bibfield  {author} {\bibinfo {author} {\bibfnamefont {Dimitris}\
  \bibnamefont {Bertsimas}}\ and\ \bibinfo {author} {\bibfnamefont {Santosh}\
  \bibnamefont {Vempala}},\ }\bibfield  {title} {\enquote {\bibinfo {title}
  {Solving convex programs by random walks},}\ }\href@noop {} {\bibfield
  {journal} {\bibinfo  {journal} {Journal of the ACM (JACM)}\ }\textbf
  {\bibinfo {volume} {51}},\ \bibinfo {pages} {540--556} (\bibinfo {year}
  {2004})}\BibitemShut {NoStop}%
\bibitem [{\citenamefont {Calafiore}\ and\ \citenamefont
  {Dabbene}(2007)}]{Calafiore2007}%
  \BibitemOpen
  \bibfield  {author} {\bibinfo {author} {\bibfnamefont {Giuseppe~C.}\
  \bibnamefont {Calafiore}}\ and\ \bibinfo {author} {\bibfnamefont {Fabrizio}\
  \bibnamefont {Dabbene}},\ }\bibfield  {title} {\enquote {\bibinfo {title} {A
  probabilistic analytic center cutting plane method for feasibility of
  uncertain lmis},}\ }\href@noop {} {\bibfield  {journal} {\bibinfo  {journal}
  {Automatica}\ }\textbf {\bibinfo {volume} {43}},\ \bibinfo {pages}
  {2022–2033} (\bibinfo {year} {2007})}\BibitemShut {NoStop}%
\bibitem [{\citenamefont {{Dabbene}}\ \emph {et~al.}(2008)\citenamefont
  {{Dabbene}}, \citenamefont {{Shcherbakov}},\ and\ \citenamefont
  {{Polyak}}}]{4739188}%
  \BibitemOpen
  \bibfield  {author} {\bibinfo {author} {\bibfnamefont {F.}~\bibnamefont
  {{Dabbene}}}, \bibinfo {author} {\bibfnamefont {P.}~\bibnamefont
  {{Shcherbakov}}}, \ and\ \bibinfo {author} {\bibfnamefont {B.~T.}\
  \bibnamefont {{Polyak}}},\ }\bibfield  {title} {\enquote {\bibinfo {title} {A
  randomized cutting plane scheme with geometric convergence: Probabilistic
  analysis and sdp applications},}\ }in\ \href {\doibase
  10.1109/CDC.2008.4739188} {\emph {\bibinfo {booktitle} {2008 47th IEEE
  Conference on Decision and Control}}}\ (\bibinfo {year} {2008})\ pp.\
  \bibinfo {pages} {3044--3049}\BibitemShut {NoStop}%
\bibitem [{\citenamefont {Dabbene}\ \emph {et~al.}(2010)\citenamefont
  {Dabbene}, \citenamefont {Shcherbakov},\ and\ \citenamefont
  {Polyak}}]{dsp10}%
  \BibitemOpen
  \bibfield  {author} {\bibinfo {author} {\bibfnamefont {Fabrizio}\
  \bibnamefont {Dabbene}}, \bibinfo {author} {\bibfnamefont {Pavel~S}\
  \bibnamefont {Shcherbakov}}, \ and\ \bibinfo {author} {\bibfnamefont
  {Boris~T}\ \bibnamefont {Polyak}},\ }\bibfield  {title} {\enquote {\bibinfo
  {title} {A randomized cutting plane method with probabilistic geometric
  convergence},}\ }\href@noop {} {\bibfield  {journal} {\bibinfo  {journal}
  {SIAM Journal on Optimization}\ }\textbf {\bibinfo {volume} {20}},\ \bibinfo
  {pages} {3185--3207} (\bibinfo {year} {2010})}\BibitemShut {NoStop}%
\bibitem [{\citenamefont {Lee}\ \emph {et~al.}(2015)\citenamefont {Lee},
  \citenamefont {Sidford},\ and\ \citenamefont {Wong}}]{lsw15}%
  \BibitemOpen
  \bibfield  {author} {\bibinfo {author} {\bibfnamefont {Yin~Tat}\ \bibnamefont
  {Lee}}, \bibinfo {author} {\bibfnamefont {Aaron}\ \bibnamefont {Sidford}}, \
  and\ \bibinfo {author} {\bibfnamefont {Sam Chiu-wai}\ \bibnamefont {Wong}},\
  }\bibfield  {title} {\enquote {\bibinfo {title} {A faster cutting plane
  method and its implications for combinatorial and convex optimization},}\
  }in\ \href@noop {} {\emph {\bibinfo {booktitle} {2015 IEEE 56th Annual
  Symposium on Foundations of Computer Science}}}\ (\bibinfo {organization}
  {IEEE},\ \bibinfo {year} {2015})\ pp.\ \bibinfo {pages}
  {1049--1065}\BibitemShut {NoStop}%
\bibitem [{\citenamefont {Jiang}\ \emph
  {et~al.}(2020{\natexlab{b}})\citenamefont {Jiang}, \citenamefont {Lee},
  \citenamefont {Song},\ and\ \citenamefont {Wong}}]{jiang2020improved}%
  \BibitemOpen
  \bibfield  {author} {\bibinfo {author} {\bibfnamefont {Haotian}\ \bibnamefont
  {Jiang}}, \bibinfo {author} {\bibfnamefont {Yin~Tat}\ \bibnamefont {Lee}},
  \bibinfo {author} {\bibfnamefont {Zhao}\ \bibnamefont {Song}}, \ and\
  \bibinfo {author} {\bibfnamefont {Sam Chiu-wai}\ \bibnamefont {Wong}},\
  }\bibfield  {title} {\enquote {\bibinfo {title} {An improved cutting plane
  method for convex optimization, convex-concave games, and its
  applications},}\ }in\ \href@noop {} {\emph {\bibinfo {booktitle} {Proceedings
  of the 52nd Annual ACM SIGACT Symposium on Theory of Computing}}}\ (\bibinfo
  {year} {2020})\ pp.\ \bibinfo {pages} {944--953}\BibitemShut {NoStop}%
\bibitem [{\citenamefont {Levin}(1965)}]{levin1965algorithm}%
  \BibitemOpen
  \bibfield  {author} {\bibinfo {author} {\bibfnamefont {Anatoly~Yur'evich}\
  \bibnamefont {Levin}},\ }\bibfield  {title} {\enquote {\bibinfo {title} {An
  algorithm for minimizing convex functions},}\ }\href@noop {} {\bibfield
  {journal} {\bibinfo  {journal} {Doklady Akademii Nauk}\ }\textbf {\bibinfo
  {volume} {160}},\ \bibinfo {pages} {1244--1247} (\bibinfo {year}
  {1965})}\BibitemShut {NoStop}%
\bibitem [{\citenamefont {Newman}(1965)}]{newman1965location}%
  \BibitemOpen
  \bibfield  {author} {\bibinfo {author} {\bibfnamefont {Donald~J}\
  \bibnamefont {Newman}},\ }\bibfield  {title} {\enquote {\bibinfo {title}
  {Location of the maximum on unimodal surfaces},}\ }\href@noop {} {\bibfield
  {journal} {\bibinfo  {journal} {Journal of the ACM (JACM)}\ }\textbf
  {\bibinfo {volume} {12}},\ \bibinfo {pages} {395--398} (\bibinfo {year}
  {1965})}\BibitemShut {NoStop}%
\bibitem [{\citenamefont {Gr{\"u}nbaum}(1960)}]{grunbaum1960partitions}%
  \BibitemOpen
  \bibfield  {author} {\bibinfo {author} {\bibfnamefont {Branko}\ \bibnamefont
  {Gr{\"u}nbaum}},\ }\bibfield  {title} {\enquote {\bibinfo {title} {Partitions
  of mass-distributions and of convex bodies by hyperplanes.}}\ }\href@noop {}
  {\bibfield  {journal} {\bibinfo  {journal} {Pacific Journal of Mathematics}\
  }\textbf {\bibinfo {volume} {10}},\ \bibinfo {pages} {1257--1261} (\bibinfo
  {year} {1960})}\BibitemShut {NoStop}%
\bibitem [{Note6()}]{Note6}%
  \BibitemOpen
  \bibinfo {note} {Note that some literature \cite {4739188,dsp10} restates the
  lemma incorrectly, with $\protect \frac {1}{n}$ instead of $n$.}\BibitemShut
  {Stop}%
\bibitem [{\citenamefont {Radon}(1921)}]{radon1921mengen}%
  \BibitemOpen
  \bibfield  {author} {\bibinfo {author} {\bibfnamefont {Johann}\ \bibnamefont
  {Radon}},\ }\bibfield  {title} {\enquote {\bibinfo {title} {Mengen konvexer
  k{\"o}rper, die einen gemeinsamen punkt enthalten},}\ }\href@noop {}
  {\bibfield  {journal} {\bibinfo  {journal} {Mathematische Annalen}\ }\textbf
  {\bibinfo {volume} {83}},\ \bibinfo {pages} {113--115} (\bibinfo {year}
  {1921})}\BibitemShut {NoStop}%
\bibitem [{\citenamefont {Bubeck}(2015)}]{Bubeck2015}%
  \BibitemOpen
  \bibfield  {author} {\bibinfo {author} {\bibfnamefont {S\'{e}bastien}\
  \bibnamefont {Bubeck}},\ }\bibfield  {title} {\enquote {\bibinfo {title}
  {Convex optimization: Algorithms and complexity},}\ }\href {\doibase
  10.1561/2200000050} {\bibfield  {journal} {\bibinfo  {journal} {Found. Trends
  Mach. Learn.}\ }\textbf {\bibinfo {volume} {8}},\ \bibinfo {pages}
  {231–357} (\bibinfo {year} {2015})}\BibitemShut {NoStop}%
\bibitem [{\citenamefont {Rademacher}(2007)}]{rademacher2007approximating}%
  \BibitemOpen
  \bibfield  {author} {\bibinfo {author} {\bibfnamefont {Luis~A}\ \bibnamefont
  {Rademacher}},\ }\bibfield  {title} {\enquote {\bibinfo {title}
  {Approximating the centroid is hard},}\ }in\ \href@noop {} {\emph {\bibinfo
  {booktitle} {Proceedings of the twenty-third annual symposium on
  Computational geometry}}}\ (\bibinfo {year} {2007})\ pp.\ \bibinfo {pages}
  {302--305}\BibitemShut {NoStop}%
\bibitem [{\citenamefont {Chalkis}\ \emph {et~al.}(2020)\citenamefont
  {Chalkis}, \citenamefont {Emiris}, \citenamefont {Fisikopoulos},
  \citenamefont {Repouskos},\ and\ \citenamefont
  {Tsigaridas}}]{chalkis2020efficient}%
  \BibitemOpen
  \bibfield  {author} {\bibinfo {author} {\bibfnamefont {Apostolos}\
  \bibnamefont {Chalkis}}, \bibinfo {author} {\bibfnamefont {Ioannis}\
  \bibnamefont {Emiris}}, \bibinfo {author} {\bibfnamefont {Vissarion}\
  \bibnamefont {Fisikopoulos}}, \bibinfo {author} {\bibfnamefont {Panagiotis}\
  \bibnamefont {Repouskos}}, \ and\ \bibinfo {author} {\bibfnamefont {Elias}\
  \bibnamefont {Tsigaridas}},\ }\bibfield  {title} {\enquote {\bibinfo {title}
  {Efficient sampling from feasible sets of sdps and volume approximation},}\
  }\href@noop {} {\bibfield  {journal} {\bibinfo  {journal} {arXiv preprint
  arXiv:2010.03817}\ } (\bibinfo {year} {2020})}\BibitemShut {NoStop}%
\bibitem [{\citenamefont {Smith}(1984)}]{smith1984efficient}%
  \BibitemOpen
  \bibfield  {author} {\bibinfo {author} {\bibfnamefont {Robert~L}\
  \bibnamefont {Smith}},\ }\bibfield  {title} {\enquote {\bibinfo {title}
  {Efficient monte carlo procedures for generating points uniformly distributed
  over bounded regions},}\ }\href@noop {} {\bibfield  {journal} {\bibinfo
  {journal} {Operations Research}\ }\textbf {\bibinfo {volume} {32}},\ \bibinfo
  {pages} {1296--1308} (\bibinfo {year} {1984})}\BibitemShut {NoStop}%
\bibitem [{\citenamefont {Lov{\'a}sz}(1999)}]{lovasz1999hit}%
  \BibitemOpen
  \bibfield  {author} {\bibinfo {author} {\bibfnamefont {L{\'a}szl{\'o}}\
  \bibnamefont {Lov{\'a}sz}},\ }\bibfield  {title} {\enquote {\bibinfo {title}
  {Hit-and-run mixes fast},}\ }\href@noop {} {\bibfield  {journal} {\bibinfo
  {journal} {Mathematical Programming}\ }\textbf {\bibinfo {volume} {86}},\
  \bibinfo {pages} {443--461} (\bibinfo {year} {1999})}\BibitemShut {NoStop}%
\bibitem [{\citenamefont {Polyak}\ and\ \citenamefont
  {Gryazina}(2014)}]{Polyak}%
  \BibitemOpen
  \bibfield  {author} {\bibinfo {author} {\bibfnamefont {Boris~T.}\
  \bibnamefont {Polyak}}\ and\ \bibinfo {author} {\bibfnamefont {E.~N.}\
  \bibnamefont {Gryazina}},\ }\bibfield  {title} {\enquote {\bibinfo {title}
  {Billiard walk-a new sampling algorithm for control and optimization},}\
  }\href@noop {} {\bibfield  {journal} {\bibinfo  {journal} {IFAC Proceedings
  Volumes}\ }\textbf {\bibinfo {volume} {47}},\ \bibinfo {pages} {6123--6128}
  (\bibinfo {year} {2014})}\BibitemShut {NoStop}%
\bibitem [{\citenamefont {Duane}\ \emph {et~al.}(1987)\citenamefont {Duane},
  \citenamefont {Kennedy}, \citenamefont {Pendleton},\ and\ \citenamefont
  {Roweth}}]{DUANE1987}%
  \BibitemOpen
  \bibfield  {author} {\bibinfo {author} {\bibfnamefont {Simon}\ \bibnamefont
  {Duane}}, \bibinfo {author} {\bibfnamefont {A.D.}\ \bibnamefont {Kennedy}},
  \bibinfo {author} {\bibfnamefont {Brian~J.}\ \bibnamefont {Pendleton}}, \
  and\ \bibinfo {author} {\bibfnamefont {Duncan}\ \bibnamefont {Roweth}},\
  }\bibfield  {title} {\enquote {\bibinfo {title} {Hybrid monte carlo},}\
  }\href@noop {} {\bibfield  {journal} {\bibinfo  {journal} {Physics Letters
  B}\ }\textbf {\bibinfo {volume} {195}},\ \bibinfo {pages} {216 -- 222}
  (\bibinfo {year} {1987})}\BibitemShut {NoStop}%
\bibitem [{\citenamefont {Mohasel~Afshar}\ and\ \citenamefont
  {Domke}(2015)}]{Afshar2015}%
  \BibitemOpen
  \bibfield  {author} {\bibinfo {author} {\bibfnamefont {Hadi}\ \bibnamefont
  {Mohasel~Afshar}}\ and\ \bibinfo {author} {\bibfnamefont {Justin}\
  \bibnamefont {Domke}},\ }\bibfield  {title} {\enquote {\bibinfo {title}
  {Reflection, refraction, and hamiltonian monte carlo},}\ }in\ \href
  {https://proceedings.neurips.cc/paper/2015/file/8303a79b1e19a194f1875981be5bdb6f-Paper.pdf}
  {\emph {\bibinfo {booktitle} {Advances in Neural Information Processing
  Systems}}},\ Vol.~\bibinfo {volume} {28},\ \bibinfo {editor} {edited by\
  \bibinfo {editor} {\bibfnamefont {C.}~\bibnamefont {Cortes}}, \bibinfo
  {editor} {\bibfnamefont {N.}~\bibnamefont {Lawrence}}, \bibinfo {editor}
  {\bibfnamefont {D.}~\bibnamefont {Lee}}, \bibinfo {editor} {\bibfnamefont
  {M.}~\bibnamefont {Sugiyama}}, \ and\ \bibinfo {editor} {\bibfnamefont
  {R.}~\bibnamefont {Garnett}}}\ (\bibinfo  {publisher} {Curran Associates,
  Inc.},\ \bibinfo {year} {2015})\ pp.\ \bibinfo {pages}
  {3007--3015}\BibitemShut {NoStop}%
\bibitem [{\citenamefont {Chevallier}\ \emph {et~al.}(2018)\citenamefont
  {Chevallier}, \citenamefont {Pion},\ and\ \citenamefont
  {Cazals}}]{chevallier2018}%
  \BibitemOpen
  \bibfield  {author} {\bibinfo {author} {\bibfnamefont {Augustin}\
  \bibnamefont {Chevallier}}, \bibinfo {author} {\bibfnamefont {Sylvain}\
  \bibnamefont {Pion}}, \ and\ \bibinfo {author} {\bibfnamefont
  {Fr{\'e}d{\'e}ric}\ \bibnamefont {Cazals}},\ }\href
  {https://hal.archives-ouvertes.fr/hal-01919855} {\emph {\bibinfo {title}
  {{Hamiltonian Monte Carlo with boundary reflections, and application to
  polytope volume calculations}}}},\ \bibinfo {type} {Research Report}\
  \bibinfo {number} {RR-9222}\ (\bibinfo  {institution} {{INRIA Sophia
  Antipolis, France}},\ \bibinfo {year} {2018})\BibitemShut {NoStop}%
\bibitem [{\citenamefont {Chevallier}(2019)}]{chevallier}%
  \BibitemOpen
  \bibfield  {author} {\bibinfo {author} {\bibfnamefont {Augustin}\
  \bibnamefont {Chevallier}},\ }\emph {\bibinfo {title} {{Random walks for
  estimating densities of states and the volume of convex bodies in high
  dimensional spaces}}},\ \href {https://tel.archives-ouvertes.fr/tel-02490412}
  {\bibinfo {type} {Theses}},\ \bibinfo  {school} {{Universit{\'e} C{\^o}te
  d'Azur}} (\bibinfo {year} {2019})\BibitemShut {NoStop}%
\bibitem [{\citenamefont {Calafiore}(2004)}]{calafiore2004random}%
  \BibitemOpen
  \bibfield  {author} {\bibinfo {author} {\bibfnamefont {Giuseppe}\
  \bibnamefont {Calafiore}},\ }\bibfield  {title} {\enquote {\bibinfo {title}
  {Random walks for probabilistic robustness},}\ }in\ \href@noop {} {\emph
  {\bibinfo {booktitle} {2004 43rd IEEE Conference on Decision and Control
  (CDC)(IEEE Cat. No. 04CH37601)}}},\ Vol.~\bibinfo {volume} {5}\ (\bibinfo
  {organization} {IEEE},\ \bibinfo {year} {2004})\ pp.\ \bibinfo {pages}
  {5316--5321}\BibitemShut {NoStop}%
\bibitem [{\citenamefont {Polyak}\ and\ \citenamefont
  {Shcherbakov}(2006)}]{polyak2006d}%
  \BibitemOpen
  \bibfield  {author} {\bibinfo {author} {\bibfnamefont {Boris~Teodorovich}\
  \bibnamefont {Polyak}}\ and\ \bibinfo {author} {\bibfnamefont
  {Pavel~Sergeevich}\ \bibnamefont {Shcherbakov}},\ }\bibfield  {title}
  {\enquote {\bibinfo {title} {The d-decomposition technique for linear matrix
  inequalities},}\ }\href@noop {} {\bibfield  {journal} {\bibinfo  {journal}
  {Automation and Remote Control}\ }\textbf {\bibinfo {volume} {67}},\ \bibinfo
  {pages} {1847--1861} (\bibinfo {year} {2006})}\BibitemShut {NoStop}%
\bibitem [{\citenamefont {Lov{\'a}sz}\ and\ \citenamefont
  {Vempala}(2006)}]{lovasz2006fast}%
  \BibitemOpen
  \bibfield  {author} {\bibinfo {author} {\bibfnamefont {L{\'a}szl{\'o}}\
  \bibnamefont {Lov{\'a}sz}}\ and\ \bibinfo {author} {\bibfnamefont {Santosh}\
  \bibnamefont {Vempala}},\ }\bibfield  {title} {\enquote {\bibinfo {title}
  {Fast algorithms for logconcave functions: Sampling, rounding, integration
  and optimization},}\ }in\ \href@noop {} {\emph {\bibinfo {booktitle} {2006
  47th Annual IEEE Symposium on Foundations of Computer Science (FOCS'06)}}}\
  (\bibinfo {organization} {IEEE},\ \bibinfo {year} {2006})\ pp.\ \bibinfo
  {pages} {57--68}\BibitemShut {NoStop}%
\bibitem [{\citenamefont {Vempala}(2005)}]{vempala2005geometric}%
  \BibitemOpen
  \bibfield  {author} {\bibinfo {author} {\bibfnamefont {Santosh}\ \bibnamefont
  {Vempala}},\ }\bibfield  {title} {\enquote {\bibinfo {title} {Geometric
  random walks: a survey},}\ }\href@noop {} {\bibfield  {journal} {\bibinfo
  {journal} {Combinatorial and computational geometry}\ }\textbf {\bibinfo
  {volume} {52}},\ \bibinfo {pages} {2} (\bibinfo {year} {2005})}\BibitemShut
  {NoStop}%
\bibitem [{\citenamefont {Lucas}(2004)}]{Lucas2004}%
  \BibitemOpen
  \bibfield  {author} {\bibinfo {author} {\bibfnamefont {Craig}\ \bibnamefont
  {Lucas}},\ }\emph {\bibinfo {title} {Algorithms for Cholesky and QR
  Factorizations, and the Semidefinite Generalized Eigenvalue Problem}},\ \href
  {http://eprints.maths.manchester.ac.uk/684/} {Ph.D. thesis},\ \bibinfo
  {school} {The University of Manchester} (\bibinfo {year} {2004})\BibitemShut
  {NoStop}%
\bibitem [{\citenamefont {Tisseur}(2000)}]{coise2000backward}%
  \BibitemOpen
  \bibfield  {author} {\bibinfo {author} {\bibfnamefont {Francoise}\
  \bibnamefont {Tisseur}},\ }\bibfield  {title} {\enquote {\bibinfo {title}
  {Backward error and condition of polynomial eigenvalue problems},}\
  }\href@noop {} {\bibfield  {journal} {\bibinfo  {journal} {Linear Algebra and
  Appl}\ }\textbf {\bibinfo {volume} {309}},\ \bibinfo {pages} {339--361}
  (\bibinfo {year} {2000})}\BibitemShut {NoStop}%
\bibitem [{\citenamefont {Guettel}\ and\ \citenamefont
  {Tisseur}(2017)}]{Guttel2007}%
  \BibitemOpen
  \bibfield  {author} {\bibinfo {author} {\bibfnamefont {Stefan}\ \bibnamefont
  {Guettel}}\ and\ \bibinfo {author} {\bibfnamefont {Francoise}\ \bibnamefont
  {Tisseur}},\ }\bibfield  {title} {\enquote {\bibinfo {title} {The nonlinear
  eigenvalue problem},}\ }\href {\doibase 10.1017/S0962492917000034} {\bibfield
   {journal} {\bibinfo  {journal} {Acta Numerica}\ }\textbf {\bibinfo {volume}
  {26}},\ \bibinfo {pages} {1–94} (\bibinfo {year} {2017})}\BibitemShut
  {NoStop}%
\bibitem [{\citenamefont {Berhanu}(2005)}]{berhanu2005polynomial}%
  \BibitemOpen
  \bibfield  {author} {\bibinfo {author} {\bibfnamefont {Michael}\ \bibnamefont
  {Berhanu}},\ }\emph {\bibinfo {title} {The polynomial eigenvalue problem}},\
  \href@noop {} {Ph.D. thesis},\ \bibinfo  {school} {University of Manchester}
  (\bibinfo {year} {2005})\BibitemShut {NoStop}%
\bibitem [{\citenamefont {Armentano}\ and\ \citenamefont
  {Beltran}(2019)}]{Armentano2019}%
  \BibitemOpen
  \bibfield  {author} {\bibinfo {author} {\bibfnamefont {Diego}\ \bibnamefont
  {Armentano}}\ and\ \bibinfo {author} {\bibfnamefont {Carlos}\ \bibnamefont
  {Beltran}},\ }\bibfield  {title} {\enquote {\bibinfo {title} {The polynomial
  eigenvalue problem is well conditioned for random inputs},}\ }\href@noop {}
  {\bibfield  {journal} {\bibinfo  {journal} {SIAM Journal on Matrix Analysis
  and Applications}\ }\textbf {\bibinfo {volume} {40}},\ \bibinfo {pages}
  {175--193} (\bibinfo {year} {2019})}\BibitemShut {NoStop}%
\bibitem [{\citenamefont {Beltr{\'a}n}\ and\ \citenamefont
  {Kozhasov}(2019)}]{beltran2019real}%
  \BibitemOpen
  \bibfield  {author} {\bibinfo {author} {\bibfnamefont {Carlos}\ \bibnamefont
  {Beltr{\'a}n}}\ and\ \bibinfo {author} {\bibfnamefont {Khazhgali}\
  \bibnamefont {Kozhasov}},\ }\bibfield  {title} {\enquote {\bibinfo {title}
  {The real polynomial eigenvalue problem is well conditioned on the
  average},}\ }\href@noop {} {\bibfield  {journal} {\bibinfo  {journal}
  {Foundations of Computational Mathematics}\ ,\ \bibinfo {pages} {1--19}}
  (\bibinfo {year} {2019})}\BibitemShut {NoStop}%
\bibitem [{\citenamefont {Boissonnat}\ and\ \citenamefont
  {Teillaud}(2006)}]{boissonnat2006effective}%
  \BibitemOpen
  \bibfield  {author} {\bibinfo {author} {\bibfnamefont {Jean-Daniel}\
  \bibnamefont {Boissonnat}}\ and\ \bibinfo {author} {\bibfnamefont {Monique}\
  \bibnamefont {Teillaud}},\ }\href@noop {} {\emph {\bibinfo {title} {Effective
  computational geometry for curves and surfaces}}}\ (\bibinfo  {publisher}
  {Springer},\ \bibinfo {year} {2006})\BibitemShut {NoStop}%
\bibitem [{\citenamefont {van~den Brand}(2020)}]{van2020deterministic}%
  \BibitemOpen
  \bibfield  {author} {\bibinfo {author} {\bibfnamefont {Jan}\ \bibnamefont
  {van~den Brand}},\ }\bibfield  {title} {\enquote {\bibinfo {title} {A
  deterministic linear program solver in current matrix multiplication time},}\
  }in\ \href@noop {} {\emph {\bibinfo {booktitle} {Proceedings of the
  Fourteenth Annual ACM-SIAM Symposium on Discrete Algorithms}}}\ (\bibinfo
  {organization} {SIAM},\ \bibinfo {year} {2020})\ pp.\ \bibinfo {pages}
  {259--278}\BibitemShut {NoStop}%
\bibitem [{\citenamefont {Kitaev}(1995)}]{kitaev1995quantum}%
  \BibitemOpen
  \bibfield  {author} {\bibinfo {author} {\bibfnamefont {A~Yu}\ \bibnamefont
  {Kitaev}},\ }\bibfield  {title} {\enquote {\bibinfo {title} {Quantum
  measurements and the abelian stabilizer problem},}\ }\href@noop {} {\bibfield
   {journal} {\bibinfo  {journal} {arXiv preprint quant-ph/9511026}\ }
  (\bibinfo {year} {1995})}\BibitemShut {NoStop}%
\bibitem [{\citenamefont {Parker}\ and\ \citenamefont
  {Joseph}(2020)}]{parker2020quantum}%
  \BibitemOpen
  \bibfield  {author} {\bibinfo {author} {\bibfnamefont {Jeffrey~B}\
  \bibnamefont {Parker}}\ and\ \bibinfo {author} {\bibfnamefont {Ilon}\
  \bibnamefont {Joseph}},\ }\bibfield  {title} {\enquote {\bibinfo {title}
  {Quantum phase estimation for a class of generalized eigenvalue problems},}\
  }\href@noop {} {\bibfield  {journal} {\bibinfo  {journal} {Physical Review
  A}\ }\textbf {\bibinfo {volume} {102}},\ \bibinfo {pages} {022422} (\bibinfo
  {year} {2020})}\BibitemShut {NoStop}%
\bibitem [{\citenamefont {Egger}\ \emph {et~al.}(2021)\citenamefont {Egger},
  \citenamefont {Mare{\v{c}}ek},\ and\ \citenamefont
  {Woerner}}]{egger2020warmstarting}%
  \BibitemOpen
  \bibfield  {author} {\bibinfo {author} {\bibfnamefont {Daniel~J}\
  \bibnamefont {Egger}}, \bibinfo {author} {\bibfnamefont {Jakub}\ \bibnamefont
  {Mare{\v{c}}ek}}, \ and\ \bibinfo {author} {\bibfnamefont {Stefan}\
  \bibnamefont {Woerner}},\ }\bibfield  {title} {\enquote {\bibinfo {title}
  {Warm-starting quantum optimization},}\ }\href@noop {} {\bibfield  {journal}
  {\bibinfo  {journal} {Quantum}\ }\textbf {\bibinfo {volume} {5}},\ \bibinfo
  {pages} {479} (\bibinfo {year} {2021})}\BibitemShut {NoStop}%
\bibitem [{\citenamefont {{Wocjan}}\ and\ \citenamefont
  {{Zhang}}(2006)}]{Wocjan2006}%
  \BibitemOpen
  \bibfield  {author} {\bibinfo {author} {\bibfnamefont {Pawel}\ \bibnamefont
  {{Wocjan}}}\ and\ \bibinfo {author} {\bibfnamefont {Shengyu}\ \bibnamefont
  {{Zhang}}},\ }\bibfield  {title} {\enquote {\bibinfo {title} {{Several
  natural BQP-Complete problems}},}\ }\href@noop {} {\bibfield  {journal}
  {\bibinfo  {journal} {arXiv e-prints}\ ,\ \bibinfo {eid} {quant-ph/0606179}}
  (\bibinfo {year} {2006})},\ \Eprint {http://arxiv.org/abs/quant-ph/0606179}
  {arXiv:quant-ph/0606179 [quant-ph]} \BibitemShut {NoStop}%
\bibitem [{\citenamefont {Gohberg}\ \emph {et~al.}(1982)\citenamefont
  {Gohberg}, \citenamefont {Lancaster},\ and\ \citenamefont
  {Rodman}}]{gohberg1982matrix}%
  \BibitemOpen
  \bibfield  {author} {\bibinfo {author} {\bibfnamefont {I.}~\bibnamefont
  {Gohberg}}, \bibinfo {author} {\bibfnamefont {P.}~\bibnamefont {Lancaster}},
  \ and\ \bibinfo {author} {\bibfnamefont {L.}~\bibnamefont {Rodman}},\ }\href
  {https://books.google.cz/books?id=KwEItnMvwbgC} {\emph {\bibinfo {title}
  {Matrix Polynomials}}},\ Classics in Applied Mathematics\ (\bibinfo
  {publisher} {Society for Industrial and Applied Mathematics (SIAM, 3600
  Market Street, Floor 6, Philadelphia, PA 19104)},\ \bibinfo {year}
  {1982})\BibitemShut {NoStop}%
\bibitem [{\citenamefont {Mackey}\ \emph {et~al.}(2006)\citenamefont {Mackey},
  \citenamefont {Mackey}, \citenamefont {Mehl},\ and\ \citenamefont
  {Mehrmann}}]{Mackey2006}%
  \BibitemOpen
  \bibfield  {author} {\bibinfo {author} {\bibfnamefont {D.}~\bibnamefont
  {Mackey}}, \bibinfo {author} {\bibfnamefont {Niloufer}\ \bibnamefont
  {Mackey}}, \bibinfo {author} {\bibfnamefont {Christian}\ \bibnamefont
  {Mehl}}, \ and\ \bibinfo {author} {\bibfnamefont {Volker}\ \bibnamefont
  {Mehrmann}},\ }\bibfield  {title} {\enquote {\bibinfo {title} {Vector spaces
  of linearizations for matrix polynomials},}\ }\href {\doibase
  10.1137/050628350} {\bibfield  {journal} {\bibinfo  {journal} {SIAM Journal
  on Matrix Analysis and Applications}\ }\textbf {\bibinfo {volume} {28}}
  (\bibinfo {year} {2006}),\ 10.1137/050628350}\BibitemShut {NoStop}%
\bibitem [{\citenamefont {Higham}\ \emph {et~al.}(2006)\citenamefont {Higham},
  \citenamefont {Mackey},\ and\ \citenamefont
  {Tisseur}}]{higham2006conditioning}%
  \BibitemOpen
  \bibfield  {author} {\bibinfo {author} {\bibfnamefont {Nicholas~J}\
  \bibnamefont {Higham}}, \bibinfo {author} {\bibfnamefont {D~Steven}\
  \bibnamefont {Mackey}}, \ and\ \bibinfo {author} {\bibfnamefont
  {Fran{\c{c}}oise}\ \bibnamefont {Tisseur}},\ }\bibfield  {title} {\enquote
  {\bibinfo {title} {The conditioning of linearizations of matrix
  polynomials},}\ }\href@noop {} {\bibfield  {journal} {\bibinfo  {journal}
  {SIAM Journal on Matrix Analysis and Applications}\ }\textbf {\bibinfo
  {volume} {28}},\ \bibinfo {pages} {1005--1028} (\bibinfo {year}
  {2006})}\BibitemShut {NoStop}%
\bibitem [{\citenamefont {Parlett}(1971)}]{parlett1971analysis}%
  \BibitemOpen
  \bibfield  {author} {\bibinfo {author} {\bibfnamefont {Beresford~N}\
  \bibnamefont {Parlett}},\ }\bibfield  {title} {\enquote {\bibinfo {title}
  {Analysis of algorithms for reflections in bisectors},}\ }\href@noop {}
  {\bibfield  {journal} {\bibinfo  {journal} {SIAM Review}\ }\textbf {\bibinfo
  {volume} {13}},\ \bibinfo {pages} {197--208} (\bibinfo {year}
  {1971})}\BibitemShut {NoStop}%
\bibitem [{\citenamefont {Fix}\ and\ \citenamefont
  {Heiberger}(1972)}]{fix1972algorithm}%
  \BibitemOpen
  \bibfield  {author} {\bibinfo {author} {\bibfnamefont {George}\ \bibnamefont
  {Fix}}\ and\ \bibinfo {author} {\bibfnamefont {Richard}\ \bibnamefont
  {Heiberger}},\ }\bibfield  {title} {\enquote {\bibinfo {title} {An algorithm
  for the ill-conditioned generalized eigenvalue problem},}\ }\href@noop {}
  {\bibfield  {journal} {\bibinfo  {journal} {SIAM Journal on Numerical
  Analysis}\ }\textbf {\bibinfo {volume} {9}},\ \bibinfo {pages} {78--88}
  (\bibinfo {year} {1972})}\BibitemShut {NoStop}%
\bibitem [{\citenamefont {Bunse-Gerstner}(1984)}]{bunse1984algorithm}%
  \BibitemOpen
  \bibfield  {author} {\bibinfo {author} {\bibfnamefont {Angelika}\
  \bibnamefont {Bunse-Gerstner}},\ }\bibfield  {title} {\enquote {\bibinfo
  {title} {An algorithm for the symmetric generalized eigenvalue problem},}\
  }\href@noop {} {\bibfield  {journal} {\bibinfo  {journal} {Linear Algebra and
  its Applications}\ }\textbf {\bibinfo {volume} {58}},\ \bibinfo {pages}
  {43--68} (\bibinfo {year} {1984})}\BibitemShut {NoStop}%
\bibitem [{\citenamefont {Cao}(1987)}]{cao1987deflation}%
  \BibitemOpen
  \bibfield  {author} {\bibinfo {author} {\bibfnamefont {Zhi-hao}\ \bibnamefont
  {Cao}},\ }\bibfield  {title} {\enquote {\bibinfo {title} {On a deflation
  method for the symmetric generalized eigenvalue problem},}\ }\href@noop {}
  {\bibfield  {journal} {\bibinfo  {journal} {Linear Algebra and its
  Applications}\ }\textbf {\bibinfo {volume} {92}},\ \bibinfo {pages}
  {187--196} (\bibinfo {year} {1987})}\BibitemShut {NoStop}%
\bibitem [{\citenamefont {Demmel}\ and\ \citenamefont
  {K{\aa}gstr{\"o}m}(1993)}]{demmel1993generalized}%
  \BibitemOpen
  \bibfield  {author} {\bibinfo {author} {\bibfnamefont {James}\ \bibnamefont
  {Demmel}}\ and\ \bibinfo {author} {\bibfnamefont {Bo}~\bibnamefont
  {K{\aa}gstr{\"o}m}},\ }\bibfield  {title} {\enquote {\bibinfo {title} {The
  generalized schur decomposition of an arbitrary pencil a--$\lambda$b—robust
  software with error bounds and applications. part i: theory and
  algorithms},}\ }\href@noop {} {\bibfield  {journal} {\bibinfo  {journal} {ACM
  Transactions on Mathematical Software (TOMS)}\ }\textbf {\bibinfo {volume}
  {19}},\ \bibinfo {pages} {160--174} (\bibinfo {year} {1993})}\BibitemShut
  {NoStop}%
\bibitem [{Note7()}]{Note7}%
  \BibitemOpen
  \bibinfo {note} {Essentially, this depends on the rank-revealing
  decomposition. Spectral decomposition required $4n^{3}$ flops, while Cholesky
  or LDL$^{T}$ require $n^3/3$ flops.}\BibitemShut {Stop}%
\bibitem [{\citenamefont {{H.~Abraham et al.}}(2019)}]{Qiskit}%
  \BibitemOpen
  \bibfield  {author} {\bibinfo {author} {\bibnamefont {{H.~Abraham et al.}}},\
  }\href {https://doi.org/10.5281/zenodo.2562111} {\enquote {\bibinfo {title}
  {Qiskit: An open-source framework for quantum computing},}\ } (\bibinfo
  {year} {2019})\BibitemShut {NoStop}%
\bibitem [{\citenamefont {Borchers}(1999)}]{borchers1999sdplib}%
  \BibitemOpen
  \bibfield  {author} {\bibinfo {author} {\bibfnamefont {Brian}\ \bibnamefont
  {Borchers}},\ }\bibfield  {title} {\enquote {\bibinfo {title} {Sdplib 1.2, a
  library of semidefinite programming test problems},}\ }\href@noop {}
  {\bibfield  {journal} {\bibinfo  {journal} {Optimization Methods and
  Software}\ }\textbf {\bibinfo {volume} {11}},\ \bibinfo {pages} {683--690}
  (\bibinfo {year} {1999})}\BibitemShut {NoStop}%
\bibitem [{\citenamefont {Mittelmann}(2003)}]{mittelmann2003independent}%
  \BibitemOpen
  \bibfield  {author} {\bibinfo {author} {\bibfnamefont {Hans~D}\ \bibnamefont
  {Mittelmann}},\ }\bibfield  {title} {\enquote {\bibinfo {title} {An
  independent benchmarking of sdp and socp solvers},}\ }\href@noop {}
  {\bibfield  {journal} {\bibinfo  {journal} {Mathematical Programming}\
  }\textbf {\bibinfo {volume} {95}},\ \bibinfo {pages} {407--430} (\bibinfo
  {year} {2003})}\BibitemShut {NoStop}%
\bibitem [{\citenamefont {O'Donoghue}\ \emph {et~al.}(2019)\citenamefont
  {O'Donoghue}, \citenamefont {Chu}, \citenamefont {Parikh},\ and\
  \citenamefont {Boyd}}]{scs}%
  \BibitemOpen
  \bibfield  {author} {\bibinfo {author} {\bibfnamefont {B.}~\bibnamefont
  {O'Donoghue}}, \bibinfo {author} {\bibfnamefont {E.}~\bibnamefont {Chu}},
  \bibinfo {author} {\bibfnamefont {N.}~\bibnamefont {Parikh}}, \ and\ \bibinfo
  {author} {\bibfnamefont {S.}~\bibnamefont {Boyd}},\ }\href@noop {} {\enquote
  {\bibinfo {title} {{SCS}: Splitting conic solver, version 2.1.4},}\ }\bibinfo
  {howpublished} {\url{https://github.com/cvxgrp/scs}} (\bibinfo {year}
  {2019})\BibitemShut {NoStop}%
\bibitem [{\citenamefont {O'Donoghue}\ \emph {et~al.}(2016)\citenamefont
  {O'Donoghue}, \citenamefont {Chu}, \citenamefont {Parikh},\ and\
  \citenamefont {Boyd}}]{scs16}%
  \BibitemOpen
  \bibfield  {author} {\bibinfo {author} {\bibfnamefont {B.}~\bibnamefont
  {O'Donoghue}}, \bibinfo {author} {\bibfnamefont {E.}~\bibnamefont {Chu}},
  \bibinfo {author} {\bibfnamefont {N.}~\bibnamefont {Parikh}}, \ and\ \bibinfo
  {author} {\bibfnamefont {S.}~\bibnamefont {Boyd}},\ }\bibfield  {title}
  {\enquote {\bibinfo {title} {Conic optimization via operator splitting and
  homogeneous self-dual embedding},}\ }\href
  {http://stanford.edu/~boyd/papers/scs.html} {\bibfield  {journal} {\bibinfo
  {journal} {Journal of Optimization Theory and Applications}\ }\textbf
  {\bibinfo {volume} {169}},\ \bibinfo {pages} {1042--1068} (\bibinfo {year}
  {2016})}\BibitemShut {NoStop}%
\bibitem [{\citenamefont {Somma}\ and\ \citenamefont
  {Boixo}(2013)}]{somma2013spectral}%
  \BibitemOpen
  \bibfield  {author} {\bibinfo {author} {\bibfnamefont {Rolando~D}\
  \bibnamefont {Somma}}\ and\ \bibinfo {author} {\bibfnamefont {Sergio}\
  \bibnamefont {Boixo}},\ }\bibfield  {title} {\enquote {\bibinfo {title}
  {Spectral gap amplification},}\ }\href@noop {} {\bibfield  {journal}
  {\bibinfo  {journal} {SIAM Journal on Computing}\ }\textbf {\bibinfo {volume}
  {42}},\ \bibinfo {pages} {593--610} (\bibinfo {year} {2013})}\BibitemShut
  {NoStop}%
\bibitem [{Note8()}]{Note8}%
  \BibitemOpen
  \bibinfo {note} {An exponential quantum speed-up claimed by \cite
  {lloyd2014quantum} only under very particular circumstances, incl. low-rank
  matrices and strong assumptions on the initialization, has since been
  disputed \cite {Tang2018,Tang2020,Tang2021,chepurko2020quantum}. We do
  \protect \emph {not} claim an exponential quantum speed-up is
  available.}\BibitemShut {Stop}%
\bibitem [{\citenamefont {Lloyd}\ \emph {et~al.}(2014)\citenamefont {Lloyd},
  \citenamefont {Mohseni},\ and\ \citenamefont
  {Rebentrost}}]{lloyd2014quantum}%
  \BibitemOpen
  \bibfield  {author} {\bibinfo {author} {\bibfnamefont {Seth}\ \bibnamefont
  {Lloyd}}, \bibinfo {author} {\bibfnamefont {Masoud}\ \bibnamefont {Mohseni}},
  \ and\ \bibinfo {author} {\bibfnamefont {Patrick}\ \bibnamefont
  {Rebentrost}},\ }\bibfield  {title} {\enquote {\bibinfo {title} {Quantum
  principal component analysis},}\ }\href@noop {} {\bibfield  {journal}
  {\bibinfo  {journal} {Nature Physics}\ }\textbf {\bibinfo {volume} {10}},\
  \bibinfo {pages} {631--633} (\bibinfo {year} {2014})}\BibitemShut {NoStop}%
\bibitem [{\citenamefont {Harrow}\ \emph {et~al.}(2009)\citenamefont {Harrow},
  \citenamefont {Hassidim},\ and\ \citenamefont {Lloyd}}]{harrow2009quantum}%
  \BibitemOpen
  \bibfield  {author} {\bibinfo {author} {\bibfnamefont {Aram~W}\ \bibnamefont
  {Harrow}}, \bibinfo {author} {\bibfnamefont {Avinatan}\ \bibnamefont
  {Hassidim}}, \ and\ \bibinfo {author} {\bibfnamefont {Seth}\ \bibnamefont
  {Lloyd}},\ }\bibfield  {title} {\enquote {\bibinfo {title} {Quantum algorithm
  for linear systems of equations},}\ }\href@noop {} {\bibfield  {journal}
  {\bibinfo  {journal} {Physical Review Letters}\ }\textbf {\bibinfo {volume}
  {103}},\ \bibinfo {pages} {150502} (\bibinfo {year} {2009})}\BibitemShut
  {NoStop}%
\bibitem [{\citenamefont {Tang}(2018)}]{Tang2018}%
  \BibitemOpen
  \bibfield  {author} {\bibinfo {author} {\bibfnamefont {Ewin}\ \bibnamefont
  {Tang}},\ }\bibfield  {title} {\enquote {\bibinfo {title} {Quantum-inspired
  classical algorithms for principal component analysis and supervised
  clustering},}\ }\href {http://arxiv.org/abs/1811.00414} {\bibfield  {journal}
  {\bibinfo  {journal} {CoRR}\ }\textbf {\bibinfo {volume} {abs/1811.00414}}
  (\bibinfo {year} {2018})},\ \Eprint {http://arxiv.org/abs/1811.00414}
  {arXiv:1811.00414} \BibitemShut {NoStop}%
\bibitem [{\citenamefont {Chia}\ \emph {et~al.}(2020)\citenamefont {Chia},
  \citenamefont {Gily\'{e}n}, \citenamefont {Li}, \citenamefont {Lin},
  \citenamefont {Tang},\ and\ \citenamefont {Wang}}]{Tang2020}%
  \BibitemOpen
  \bibfield  {author} {\bibinfo {author} {\bibfnamefont {Nai-Hui}\ \bibnamefont
  {Chia}}, \bibinfo {author} {\bibfnamefont {Andr\'{a}s}\ \bibnamefont
  {Gily\'{e}n}}, \bibinfo {author} {\bibfnamefont {Tongyang}\ \bibnamefont
  {Li}}, \bibinfo {author} {\bibfnamefont {Han-Hsuan}\ \bibnamefont {Lin}},
  \bibinfo {author} {\bibfnamefont {Ewin}\ \bibnamefont {Tang}}, \ and\
  \bibinfo {author} {\bibfnamefont {Chunhao}\ \bibnamefont {Wang}},\ }\bibfield
   {title} {\enquote {\bibinfo {title} {Sampling-based sublinear low-rank
  matrix arithmetic framework for dequantizing quantum machine learning},}\
  }in\ \href {\doibase 10.1145/3357713.3384314} {\emph {\bibinfo {booktitle}
  {Proceedings of the 52nd Annual ACM SIGACT Symposium on Theory of
  Computing}}},\ \bibinfo {series and number} {STOC 2020}\ (\bibinfo
  {publisher} {Association for Computing Machinery},\ \bibinfo {address} {New
  York, NY, USA},\ \bibinfo {year} {2020})\ p.\ \bibinfo {pages}
  {387–400}\BibitemShut {NoStop}%
\bibitem [{\citenamefont {Tang}(2021)}]{Tang2021}%
  \BibitemOpen
  \bibfield  {author} {\bibinfo {author} {\bibfnamefont {Ewin}\ \bibnamefont
  {Tang}},\ }\bibfield  {title} {\enquote {\bibinfo {title} {Quantum principal
  component analysis only achieves an exponential speedup because of its state
  preparation assumptions},}\ }\href {\doibase 10.1103/PhysRevLett.127.060503}
  {\bibfield  {journal} {\bibinfo  {journal} {Phys. Rev. Lett.}\ }\textbf
  {\bibinfo {volume} {127}},\ \bibinfo {pages} {060503} (\bibinfo {year}
  {2021})}\BibitemShut {NoStop}%
\bibitem [{\citenamefont {Chepurko}\ \emph {et~al.}(2020)\citenamefont
  {Chepurko}, \citenamefont {Clarkson}, \citenamefont {Horesh},\ and\
  \citenamefont {Woodruff}}]{chepurko2020quantum}%
  \BibitemOpen
  \bibfield  {author} {\bibinfo {author} {\bibfnamefont {Nadiia}\ \bibnamefont
  {Chepurko}}, \bibinfo {author} {\bibfnamefont {Kenneth~L}\ \bibnamefont
  {Clarkson}}, \bibinfo {author} {\bibfnamefont {Lior}\ \bibnamefont {Horesh}},
  \ and\ \bibinfo {author} {\bibfnamefont {David~P}\ \bibnamefont {Woodruff}},\
  }\bibfield  {title} {\enquote {\bibinfo {title} {Quantum-inspired algorithms
  from randomized numerical linear algebra},}\ }\href@noop {} {\bibfield
  {journal} {\bibinfo  {journal} {arXiv preprint arXiv:2011.04125}\ } (\bibinfo
  {year} {2020})}\BibitemShut {NoStop}%
\bibitem [{\citenamefont {Alzer}(2003)}]{alzer03}%
  \BibitemOpen
  \bibfield  {author} {\bibinfo {author} {\bibfnamefont {Horst}\ \bibnamefont
  {Alzer}},\ }\bibfield  {title} {\enquote {\bibinfo {title} {Some
  beta-function inequalities},}\ }\href@noop {} {\bibfield  {journal} {\bibinfo
   {journal} {Proceedings. Section A, Mathematics-The Royal Society of
  Edinburgh}\ }\textbf {\bibinfo {volume} {133}},\ \bibinfo {pages} {731}
  (\bibinfo {year} {2003})}\BibitemShut {NoStop}%
\bibitem [{Note9()}]{Note9}%
  \BibitemOpen
  \bibinfo {note} {Note that unlike some analyses in the literature \cite
  {4739188,dsp10}, we use the correct $h^2$ instead of $h$.}\BibitemShut
  {Stop}%
\bibitem [{Note10()}]{Note10}%
  \BibitemOpen
  \bibinfo {note} {Depending on what data are used for the numerical
  experiments, we may not need to solve the initial optimization problem. For
  example, if we choose $F_0$ such that $F_0 \prec 0$, then we can use $x_0 =
  0$ as the initial point.}\BibitemShut {Stop}%
\end{thebibliography}%

\clearpage
\appendix
\onecolumngrid

\section{Statistical properties of the empirical minimum over a convex body}
\label{app:statistical}

Given $c\in \mathbb{R}^n$, we define the following random variables that represent the value of a linear objective evaluated at the random points $x^{(i)}$: $f^{(i)}=c^{T}x^{(i)}$, $i=1,\ldots,N$. Then we can define the so-called \emph{empirical minimum} over these random points as
\begin{equation}
    f_{[1]} = \min_{i=1,\ldots,N}f^{(i)}\label{eq:1}. 
\end{equation}
Notice that $f_{[1]}$ is also a random variable; it represents the socalled \emph{first order statistics} of $f^{(i)}$. The key theorem below proves that, for every convex body $\mathcal{X}$, the expected value of the relative distance between the empirical minimum $f_{[1]}$ and the true one $f^* = \min_{\mathcal{X}}c^{T}x$ is bounded from below and from above by constants that depend only on $n$ and $N$. 

\begin{lemma}[Brunn]\label{lem:1} Let $\calX \subset \mathbb{R}^n$ be a convex body. Define a parallel slice $\calX_s = \calX \cap \{x: x_1 = s\}$ and its $(n-1)$-dimensional volume $v_{\calX}(s) = \vol(\calX_s)$. Then the function $v_{\calX}(s)^{\frac{1}{n-1}}$ is concave, and
$$\vol(\calX_s)^{\frac{1}{n-1}} \geq \lambda\vol(\calX_{s_1})^{\frac{1}{n-1}} + (1-\lambda)\vol(\calX_{s_2})^{\frac{1}{n-1}},$$
where $s=\lambda s_1 + (1-\lambda) s_2$, $\lambda \in [0,1]$.
\end{lemma}
\begin{proof}
First note that $\lambda\calX_{s_1} + (1-\lambda)\calX_{s_2} \subseteq \calX_s$. Indeed, using convexity of $\mathcal{X}$, for all $(s_1,\bar{y}_1) \in \calX_{s_1}$ and $(s_2,\bar{y}_2) \in \calX_{s_2}$, we have $(s,\bar{y}) = (\lambda s_1 + (1-\lambda)s_2, \lambda \bar{y}_1 + (1-\lambda)\bar{y}_2) = \lambda (s_1,\bar{y}_1)+(1-\lambda) (s_2,\bar{y}_2) \in \mathcal{X} \cap \{x:x_1=s\}=\calX_s$. Then, by the Brunn-Minkowski inequality, we have
$$\vol(\calX_s)^{\frac{1}{n-1}} \geq \vol(\lambda\calX_{s_1} + (1-\lambda)\calX_{s_2})^{\frac{1}{n-1}} \geq \lambda\vol(\calX_{s_1})^{\frac{1}{n-1}} + (1-\lambda)\vol(\calX_{s_2})^{\frac{1}{n-1}},$$
i.e., $v_{\calX}(s)^{\frac{1}{n-1}}$ is concave.
\end{proof}

\begin{theorem}\label{thm:1}
Let $\mathcal{X} \subset \mathbb{R}^n$ be a convex body. Given $c\in \mathbb{R}^n$, define $h = \max_{\mathcal{X}}c^{T}x - \min_{\mathcal{X}}c^{T}x$ and $f^* = \min_{\mathcal{X}}c^{T}x$. Then it holds that

\begin{align}
    \frac{h}{nN+1} \leq \expec{f_{[1]}-f^*} &\leq \frac{h}{n}B\left(N+1, \frac{1}{n}\right) \label{eq:3.2}\\
                                            &\leq h\left(\frac{1}{N+1}\right)^{\frac{1}{n}}\label{eq:3.3},
\end{align}
where the expectation is taken with respect to samples $x^{(1\ldots\infty)}$ and $B(\cdot,\cdot)$ is the Euler Beta function.
\end{theorem}

\begin{proof}
Assume, without loss of generality, that $c=[\,1\,0 \ldots 0\,]^{T}$ (that is, $c^{T}x = x_1$) and that $x^* = \argmin_{\mathcal{X}}c^{T}x=0$. (Indeed, if in general $c^{T}x=\sum_{i=1}^{i=n}c_i x_i$, then let $y_1 = \sum_{i=1}^{i=n}c_i x_i$, so equivalently, we have $\min \Tilde{c}^{T}y = y_1$, $y\in\Tilde{\mathcal{X}}$, where $\Tilde{c} = [\,1\,0\ldots 0\,]$ and $\Tilde{\mathcal{X}} = \{y\,|\,y_1 = \sum_{i=1}^{i=n}c_i x_i,\,x\in\mathcal{X}\}$. Notice that $\Tilde{\mathcal{X}}$ is still a convex body. If $x^* = \argmin_{\mathcal{X}}c^{T}x \neq 0$, then let $y = x-x^*$, so equivalently, we have $y^*=\argmin_{\{\mathcal{X}-x^*\}}c^{T}y=0$.) 

We begin by proving the upper bound in (\ref{eq:3.2}), following \cite{dsp10}, but correcting several flawed steps. Let
$$\calX_s = \calX \cap \{x: x_1 = s\}$$
Then by Lemma \ref{lem:1}, we have
$$\vol(\calX_s)^{\frac{1}{n-1}} \geq \lambda\vol(\calX_{s_1})^{\frac{1}{n-1}} + (1-\lambda)\vol(\calX_{s_2})^{\frac{1}{n-1}}$$
for $s=\lambda s_1 + (1-\lambda) s_2$, $\lambda \in [0,1]$. 
Now define $\calX_0$ obtained by replacing each $\calX_s$ by an $(n-1)$-dimensional ball $\calB_s$ of the same volume and centered at the point $[\,s\,0\ldots 0\,]^T$, as shown in Figure 3.1. Then $\vol(\calX_0) = \vol(\calX)$ and 
\begin{equation}\label{eq:4}
   \vol(\calB_s)^{\frac{1}{n-1}} \geq \lambda\vol(\calB_{s_1})^{\frac{1}{n-1}} + (1-\lambda)\vol(\calB_{s_2})^{\frac{1}{n-1}} 
\end{equation}
Note that the volume of a $n$-dimensional ball $\calB$ with radius $r$ is
\[
    \vol(S) = \frac{\pi^{\frac{n}{2}}}{\Gamma(\frac{n}{2}+1)} r^{n},
\]
where $\Gamma$ is Euler's gamma function.
Hence, if we denote by $r(s)$ the radius of the $(n-1)$-dimensional ball $\calB$ at $s$, then (\ref{eq:4}) implies that $r(\lambda s_1 + (1-\lambda) s_2) \geq \lambda r(s_1) + (1-\lambda)r(s_2)$. This, in turn, implies that $r(s)$ is a concave function; thus $\mathcal{X}_0$ is a convex set (intuitively, think the shape of $\mathcal{X}_0$). Note that now $\calX_0$ is symmetric about the $x_1$ axis.

As a second step, define now the cone $\mathcal{\mathcal{K}}$ with base area $S = \frac{n}{h}\vol(\mathcal{X})$, the axis directed along $c$ and located as shown in Figure 3.2, with $h$ being the height of $\mathcal{\mathcal{K}}$. Then, by construction, we have $\vol(\mathcal{\mathcal{K}})=\vol(\mathcal{X})=\vol(\mathcal{X}_0)$. Let $s^*$ be the coordinate at which the sets $\mathcal{X}_0$ and $\mathcal{\mathcal{K}}$ intersect; see Figure 3.2. Next, for every $s\in [0,h]$, define the sets $\mathcal{X}^+(s)=\{x\in\mathcal{X}:x_1\geq s\}$, $\mathcal{X}^+_0(s)=\{x\in\mathcal{X}_0:x_1\geq s\}$ and $\mathcal{K}^+(s)=\{x\in\mathcal{K}:x_1\geq s\}$ as shown in Figure 3.3. Then the following chain of inequalities holds:
\begin{align*}
    \mathbb{P}\{f^{(i)}\geq s\} &= \mathbb{P}\{x_1^{(i)}\geq s\} = \frac{\vol(\mathcal{X}^+(s))}{\vol(\mathcal{X})} = \frac{\vol(\mathcal{X}^+_0(s))}{\vol(\mathcal{X}_0)}\\
                                & \leq \frac{\vol(\mathcal{K}^+(s))}{\vol(\mathcal{K})} = 1 - \frac{\vol(\mathcal{K}^-(s))}{\vol(\mathcal{K})} = 1-\left(\frac{s}{h}\right)^{n} = \frac{h^{n}-s^{n}}{h^{n}}\numberthis \label{eq:5}
\end{align*}
where the last inequality follows from the fact that, for $s\geq s^*$,
\[
    \frac{\vol(\mathcal{X}^+_0(s))}{\vol(\mathcal{X}_0)} \leq \frac{\vol(\mathcal{K}^+(s))}{\vol(\mathcal{K})},
\]
and, for $s\leq s^*$,
\[
    \frac{\vol(\mathcal{X}^+_0(s))}{\vol(\mathcal{X}_0)} = 1 - \frac{\vol(\mathcal{X}^-_0(s))}{\vol(\mathcal{X}_0)} \leq 1 - \frac{\vol(\mathcal{K}^-(s))}{\vol(\mathcal{X}_0)} = \frac{\vol(\mathcal{K}^+(s))}{\vol(\mathcal{K})},
\]
where $\mathcal{X}_0^-(s)=\{x\in \mathcal{X}_0:x_1 < s\}$ and $\mathcal{K}^-(s)=\{x\in \mathcal{K}: x_1 < s\}$.

As a final step for proving the upper bound in (\ref{eq:3.2}), notice that $f_{[1]}$ is a positive random variable, and hence we may write
\begin{align*}
    \expec{f_{[1]}} &= \int_0^h \mathbb{P}\{f_{[1]}\geq s\}\mathrm{d}s = \int_0^h\left(\mathbb{P}\{f^{(i)}\geq s\}\right)^N \mathrm{d}s\\
    [\text{from}\ (\ref{eq:5})]&\leq \int_0^h\left(\frac{h^{n}-s^{n}}{h^{n}}\right)^N \mathrm{d}s\\
    [t=s^{n}/h^{n}]&=\frac{h}{n}\int_0^1(1-t)^N t^{\frac{1}{n}-1} dt = \frac{h}{n}B\left(N+1,\frac{1}{n}\right).
\end{align*}
Finally, applying Theorem 3.4 in \cite{alzer03}, we get the following bounds
\[
    1-\left(\frac{N}{N+1}\right)^{\frac{1}{n}} \leq \frac{1}{n}B\left(N+1,\frac{1}{n}\right) \leq \left(\frac{1}{N+1}\right)^{\frac{1}{n}},
\]
thus proving the inequality in (\ref{eq:3.3}), with the cone $\mathcal{K}$ being the attainable ``worst-case'' configuration for $\mathcal{X}$. 

The lower bound in (\ref{eq:3.2}) can be proved similarly. Namely, instead of the $\mathcal{K}$ above, consider the ``inverted'' cone; then with reasonings identical to those above, it proves to be the "best case" configuration. We hence derive the following inequality:
\begin{equation}\label{eq:6}
    \mathbb{P}\{f^{(i)}\geq s\} \geq \frac{s^{n}}{h^{n}}.
\end{equation}
Therefore, we obtain
\begin{align*}
    \expec{f_{[1]}} &= \int_0^h \mathbb{P}\{f_{[1]}\geq s\}\mathrm{d}s = \int_0^h\left(\mathbb{P}\{f^{(i)}\geq s\}\right)^N \mathrm{d}s\\
    [\text{from}\ (\ref{eq:6})] &\geq \int_0^h\left(\frac{s^{n}}{h^{n}}\right)^N \mathrm{d}s = \frac{h}{nN+1},
\end{align*}
which concludes our proof.
\end{proof}
\begin{corollary}
Let $\calX \subset \mathbb{R}$ be a convex body. Define $h$, $f^*$, and $f_{[1]}$ as in Theorem \ref{thm:1}. Then it holds that
\begin{equation}
    \expec{(f_{[1]}-f^*)^2} \leq \frac{2h^2}{n}B\left(N+1,\frac{2}{n}\right) \leq h^2\left(\frac{1}{N+1}\right)^{\frac{2}{n}}.
\end{equation}
\end{corollary}
\begin{proof}
Assume again, without loss of generality, that $f^*=0$ and $c=[\,1\,0\ldots 0\,]^T$. Then $f_{[1]}$ is a positive random variable; hence, we may write \footnote{Note that unlike some analyses in the literature \cite{4739188,dsp10}, we use the correct  $h^2$ instead of $h$.}
\begin{align*}
    \expec{f_{[1]}^2} &= \int_0^{h^2} \mathbb{P}\{f_{[1]}^2\geq s\}\mathrm{d}s = \int_0^{h^2} \mathbb{P}\{f_{[1]}\geq \sqrt{s}\}\mathrm{d}s = \int_0^{h^2}\left(\mathbb{P}\{f^{(i)}\geq \sqrt{s}\}\right)^N \mathrm{d}s\\
    [\text{from}\ (\ref{eq:5})] &\leq \int_0^{h^2}\left(\frac{h^{n}-s^{\frac{n}{2}}}{h^{n}}\right)^N \mathrm{d}s\\
    \left[t=\frac{s^{\frac{n}{2}}}{h^n}\right] &=\frac{2h^2}{n}\int_0^1(1-t)^N t^{\frac{2}{n}-1} \mathrm{d}t = \frac{2h^2}{n}B\left(N+1,\frac{2}{n}\right).
\end{align*}
\end{proof}

\section{Expected convergence rate of RCP}
\label{sec:convergencerate}

At step $k$ of the RCP algorithm, define the following random variable:
\begin{equation}
    f_k = c^T z_k.
\end{equation}
Then the following corollary of Theorem 3.1 shows that the RCP scheme converges in first and second mean and, more importantly, that the rate of convergence is exponential.

\begin{theorem}[Expected convergence of RCP]\label{cor:5.1} 
Consider the RCP algorithm with $N_k \equiv N$. Then we have

\begin{equation}\label{eq:5.2}
    \expec{f_k - f^*} \leq \left(\frac{1}{N+1}\right)^{\frac{k}{n}} \expec{f_0 - f^*}
\end{equation}
that is, the RCP algorithm conversges in mean with rate $\left(\frac{1}{N+1}\right)^{\frac{1}{n}}$. Moreover, the RCP algorithm converges also in mean square with
\begin{equation}
    \expec{(f_k - f^*)^2} \leq \left(\frac{1}{N+1}\right)^{\frac{2k}{n}} \expec{(f_0 - f^*)^2}.
\end{equation}

\end{theorem}

\begin{remark}
From inequality (\ref{eq:5.2}) in the corollary we see that the expected number of steps required by the RCP algorithm to compute an $\alpha$-optimal solution (i.e., such that $\expec{f_k - f^* \leq \alpha}$) is at most 
\[
    k = \ceil*{\frac{1}{\ln (N+1)}n\ln \frac{R}{\alpha}}, 
\]
where $R=\expec{f_0 - f^*}$. Interestingly, when $N=1$, (\ref{eq:5.2}) becomes inequality (\ref{eq:2.3}). Indeed, when $N=1$, 
$$\frac{1}{n} B\left(2,\frac{1}{n}\right) = \frac{n}{n+1}$$
and 
$$\expec{f_k} = c^T \cg(\calX_k).$$
Using (\ref{eq:3.2}) and the proof of Corollary \ref{cor:5.1} (note that $h = f_{k-1} - f^*$ in the proof), we obtain
$$\expec{f_k - f^*} \leq \frac{n}{n+1}\expec{f_{k-1} - f^*}$$ 
and if $g_k \triangleq c^T \cg(\calX_k)$ and $g^* \triangleq f^*$, we have
$$g_k - g^* \leq \frac{n}{n+1}(g_{k-1} - g^*),$$
which is nothing but (\ref{eq:2.3}). Hence, we conclude that the derived convergence rate in Corollary \ref{cor:5.1} reduces to the one in Lemma \ref{lem:2.1} when $N=1$, while it improves by a factor of $\ln (N+1)$ when $N>1$.

\end{remark}

\section{Pseudocode}
\label{sec:impl}

\renewcommand{\labelenumi}{(\alph{enumi})}
\renewcommand{\labelenumii}{\arabic{enumii}:}
\begin{enumerate}
    \item \textit{Initialization}.\footnote{Depending on what data are used for the numerical experiments, we may not need to solve the initial optimization problem. For example, if we choose $F_0$ such that $F_0 \prec 0$, then we can use $x_0 = 0$ as the initial point.} We solve the following auxiliary problem 
        \begin{equation*}
            \begin{aligned}
                \min \gamma \textrm{ s.t. } F(x) &\preceq \gamma I.
            \end{aligned}
        \end{equation*}
        Note that $\{x=0,\,\gamma=\max\eig(F(0))\}$ is a feasible solution, therefore we can solve for the optimal solution $(x^*, \gamma^*)$. If $\gamma^* > 0$, then (\ref{eq:7.1}) is infeasible; otherwise, take $x_0 = x^*$ as initial feasible point for (\ref{eq:7.1}). 
        
    \item \textit{Main algorithm} (\texttt{RCP\_H\&R})
        \begin{enumerate}
            \item \textbf{Input:} $\calX$, $M$, $N$, $x_0\in\calX$
            \item \textbf{Output:} $z_k$
            \item $k=0$, $\calX_0=\calX$, $P_0 = \{\}$, $Y_0 = I_{n\times n}$, $x^{(0)}_0=x_0$, $z_0'=x_0$
            \item \textbf{for} $j=1$ to $N$ \textbf{do} 
            \item $(x^{(j)}_k,\underline{x}^{(j)}_k, \bar{x}^{(j)}_k) = \texttt{H\&R\_SDP}(\calX, Y_k, M, x^{(j-1)}_k, z_k')$
            \item $P_k = P_k \cup \{(x^{(j)}_k,\underline{x}^{(i)}_k, \bar{x}^{(i)}_k)\}$
            \item \textbf{end for}
            \item $(z_k, \underline{z}_k, \bar{z}_k) = \argmin\limits_{(x, \underline{x}, \bar{x})\in P_k}c^T x$
            \item $(z_{k}', \underline{z}_{k}', \bar{z}_{k}') = \argmin\limits_{(x, \underline{x}, \bar{x})\in 
            P_k\backslash\{(z_k, \underline{z}_k, \bar{z}_k)\}}c^T x$ (second best minimum)
            \item $x^{(0)}_{k+1} = z_k$ (initial point for $\calX_{k+1}$: $z_k \in \calX_{k+1}$, and $z_k$ is ensured in \texttt{H\&R\_SDP} to stay in the interior of $\calX$ )
            \item Calculate affine transformation matrix $Y_{k+1}$:
            \begin{equation*}
                \bar{y} = \frac{1}{2N}\sum_{i=1}^N (\underline{z}^{(i)}_k + \bar{z}^{(i)}_k),\ Y_{k+1} = \frac{1}{2N}\sum_{i=1}^N\left[(\underline{z}^{(i)}_k-\bar{y})(\underline{z}^{(i)}_k-\bar{y})^T + (\bar{z}^{(i)}_k-\bar{y})(\bar{z}^{(i)}_k-\bar{y})^T\right]
            \end{equation*}
            \item $\mathcal{X}_{k+1} = \{x\in\mathcal{X}_k:c^T(x-z_k')\leq 0\}$
            \item check Stopping Rule; $k\Leftarrow k+1$; go to 4.
        \end{enumerate}
        
    \item \textit{H\&R Algorithm for SDP} (\texttt{H\&R\_SDP})
        \begin{enumerate}
            \item \textbf{Input:} $\calX$, $Y$, $M$, $x^{(0)}$ (starting point), $z'$ (cutting point defining the input convex body)
            \item \textbf{Output:} $(x,\underline{x}, \bar{x})$, where $\underline{x}$ and $\bar{x}$ are two intersection points on the boundary of the input convex body defined by $z'$, and $x$ is a uniformal random point drawn from the line segment $[\underline{x},\bar{x}]$
            \item $y^{(0)} = x^{(0)}$
            \item \textbf{for} $i=0$ to $M-1$ \textbf{do}
            \item generate a uniformly distributed random direction $\eta\in\mathbb{R}^n$ on the unit sphere and apply affine transformation $Y$ to obtain $v=Y^{\frac{1}{2}}\eta$
            \item $(\underline{x},\bar{x}) = \texttt{BO\_SDP}(\calX, y^{(i)}, v, z')$
            \item generate a uniformly distributed point $y^{(i+1)}$ in the line segment $[\underline{x},\bar{x}]$ 
            \item \qquad\textbf{while not} $F(y^{(i+1)}) \prec 0$ \textbf{repeat}
            \item \qquad\qquad step 5, 6, 7
            \item \qquad\textbf{end while} 
            \item (Remark: with an additional \emph{count} variable, may the above be used as a stopping criterion for \texttt{RCP\_H\&R} ?)
            \item \textbf{end for}
            \item $x=y^{(M)}$; \textbf{return} $(x,\underline{x}, \bar{x})$
            \end{enumerate}
            
    \item \textit{Boundary Oracle} (\texttt{BO})
        \begin{enumerate}
            \item \textbf{Input:} $\calX$, $y$ (current point), $v$ (random direction), $z'$ (cutting point defining the input convex body)
            \item \textbf{Output:} $(\underline{x}, \bar{x})$, the intersection points between line $z=y+\lambda v$ and the boundary of the convex body defined by $z'$
            \item $(\underline{x},\bar{x}) = \texttt{BO\_SDP}(\calX, y, v)$ (see Lemma \ref{lem:7.1}.)
            \item \textbf{if} $c^T(\underline{x} - z') > 0$ \textbf{then} solve for $\underline{\lambda}: c^T(y+\lambda v - z')=0$; $\underline{x} = y + \underline{\lambda} v$
            \item \textbf{elseif} $c^T(\bar{x} - z') > 0$ \textbf{then} solve for $\bar{\lambda}: c^T(y+\lambda v - z')=0$; $\bar{x} = y + \bar{\lambda} v$
            \item \textbf{else} break
            \item \textbf{return} $(\underline{x}, \bar{x})$
            \end{enumerate}
    \item \textit{Data} \\
    $F_i$ are generated to guarantee non-emptiness of $\calX_{\LMI}$, and $F_0$ is generated such that $F_0 \prec 0$:
    $$M = 2\rand(m) - 1;\ F_0 = -M * M^T - \eye(m)$$
        \renewcommand{\labelenumii}{\Roman{enumii}.}
        \begin{enumerate}
        \item $$M = 2\rand(\frac{m}{2})-1;\, M = M + M^T;\, F_i = \blkdiag(M;-M)$$
        \item $$M = 2\rand(\frac{m}{2})-1;\, M = \triu M + (\triu M)^T - \diag(\diag(M));F_i = \blkdiag(M;-M)$$
        \item Worst-case geometry: 
              $$\calX_{\LMI} = \{x\in\mathbb{R}^n:\norm{x}_1 \leq 1, x_1 < 0\}$$
        \end{enumerate}
\end{enumerate}


\section{Details of the Implementation}
\label{sec:details}

For simplicity we compute the initial solution using SCS solver from CVXPY package. Namely, we iterate SCS until any feasible solution emerged, does not matter how far it is from the reference value of objective function. This kind of solution serves as an initial one in all our simulations. Alternatively, one can use the approach suggested in \cite{4739188}, Section~6.1 but it would take even longer simulation time.

We model noisy quantum eigensolver by adding noise to the exactly computed eigenvalues.
We have considered two noise models.
In both models, the noise was defined by its signal to noise ratio (SNR), expressed in  dB. 
Here ``signal'' is an absolute value of generalised eigenvalue (Lemma 5 in the main text) and ``noise'' is its disturbance. In this account, small SNR means strong noise and poor estimation of eigenvalues, while high SNR implies that reliable estimates of the eigenvalues are available.

The first noise model is multiplicative. It respects the spectrum of the generalised eigenproblem in the sense that every eigenvalue $\lambda_i$ is disturbed in proportion to its amplitude:
$$
\lambda^{noisy}_i = \lambda^{exact}_i \left(1 + \frac{\varepsilon}{10^{(SNR / 20)}}\right),
\qquad \varepsilon \,{\sim}\, \mathcal{N}(0, 1).
$$

Fig.~\ref{fig:noise-model-one} shows convergence profiles for various noise levels.

\begin{figure}
\centering
\includegraphics[width=0.85\textwidth]{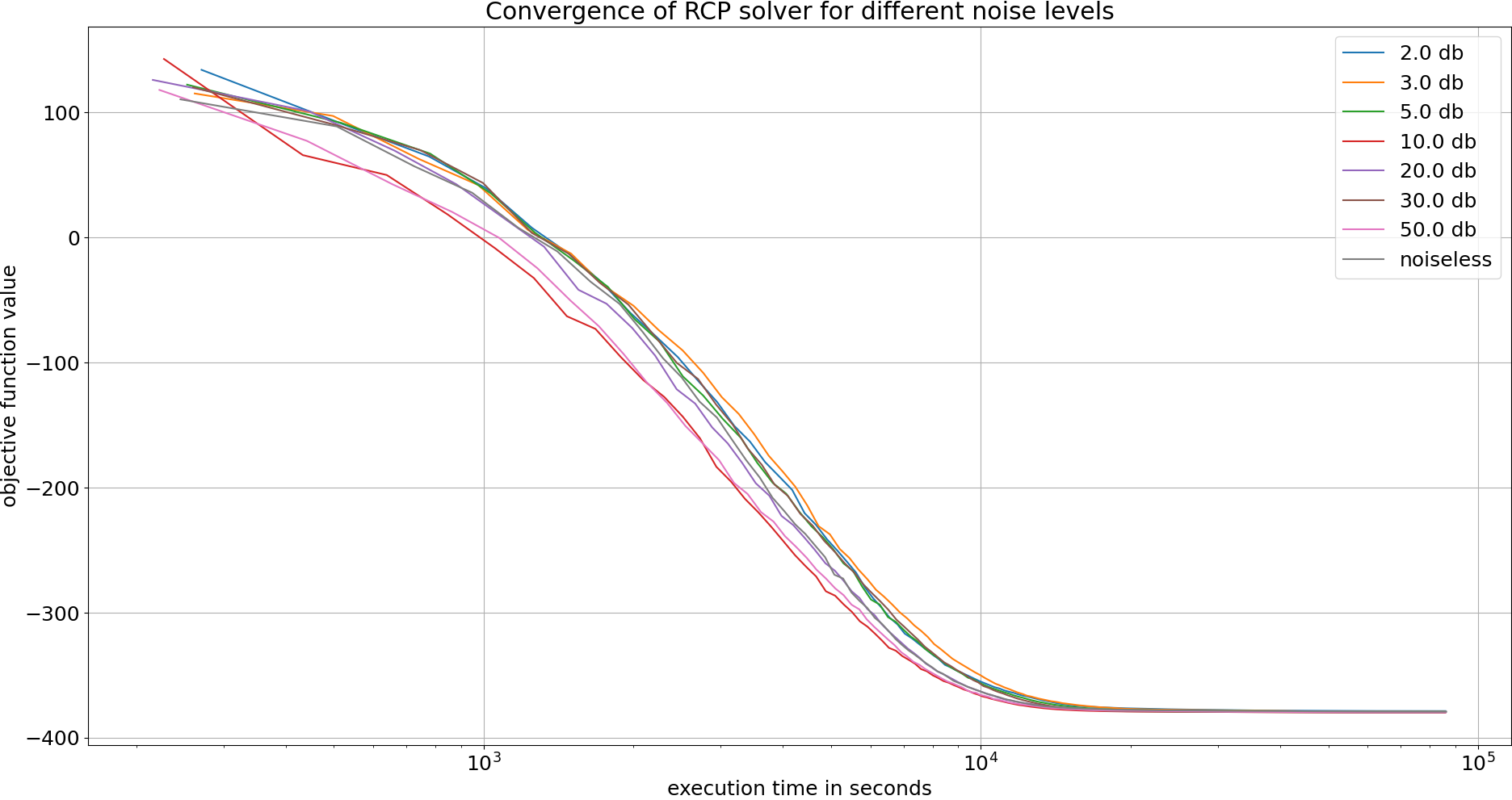}
\caption{Noise model~1 applied with different signal to noise ratios (SNR) to the problem \texttt{qap6}. The curves demonstrate evolution of objective function value in time.}
\label{fig:noise-model-one}
\end{figure}

The second noise model is additive. It adds random Gaussian noise scaled by root-mean-square eigen-value to all others:
$$
\lambda^{noisy}_i = \lambda^{exact}_i + 
\varepsilon \sqrt{\frac{\frac{1}{N} 
\sum_{i=1}^N \left(\lambda^{exact}_i\right)^2}{10^{(SNR / 10)}}},
\qquad \varepsilon \,{\sim}\, \mathcal{N}(0, 1).
$$
Mind denominator in $SNR / 10$ expression, where $10$ comes from the fact that we operate on squared ``signal'' ($\lambda_i^2$), as opposed to the first model, where just a ``signal'' amplitude is involved.  
Since eigenvalues typically differ in their orders of magnitude, the algorithm does not converge to the desired minimum ($\approx -380$ for the \texttt{qap6} problem) for high noise levels (low SNR values). The Fig.~\ref{fig:noise-model-two} demonstrates early termination in many cases.

\begin{figure}
\centering
\includegraphics[width=0.85\textwidth]{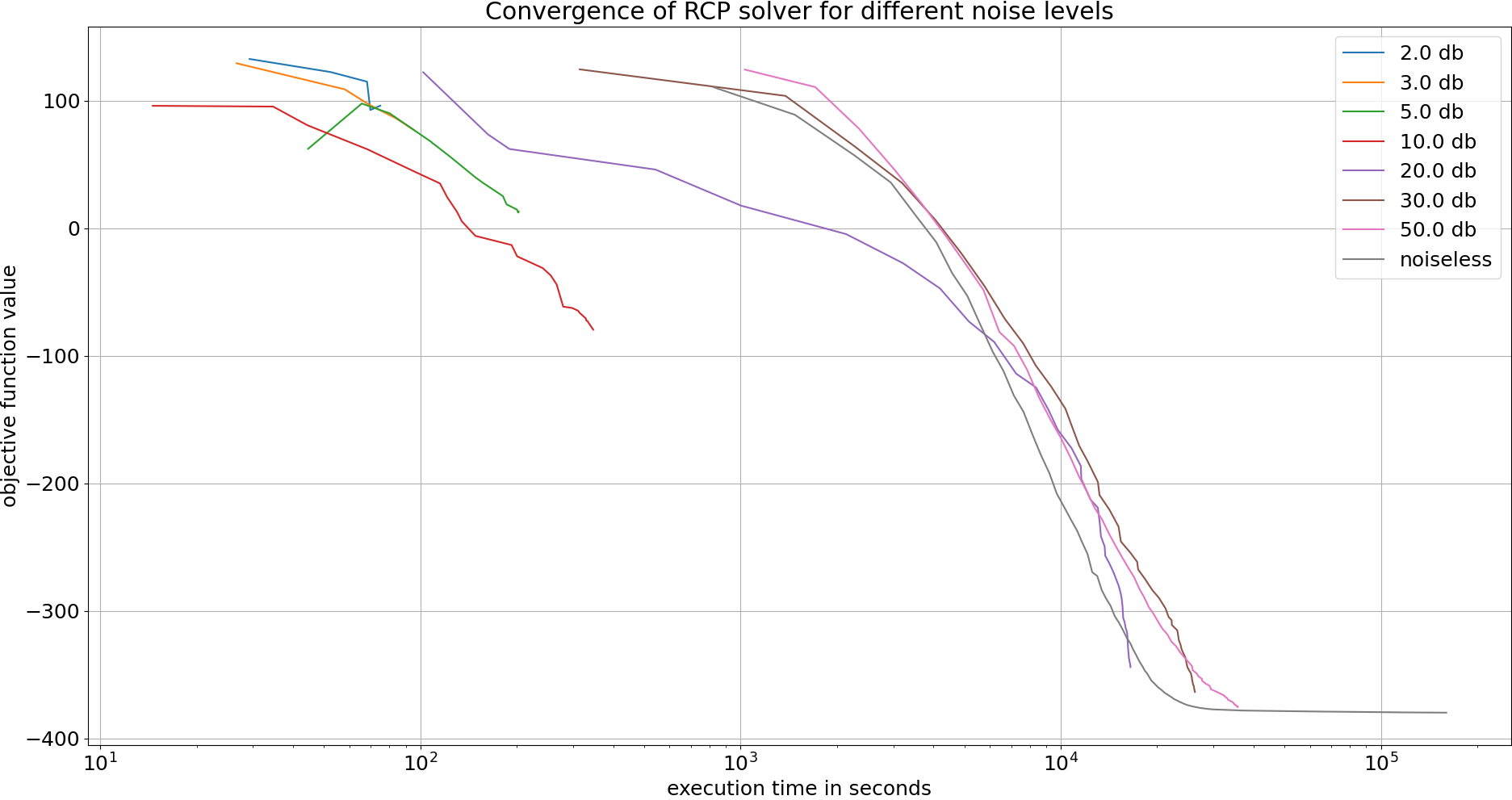}
\caption{Noise model~2 applied with different signal to noise ratios (SNR) to the problem \texttt{qap6}. The curves demonstrate evolution of objective function value in time. Noisy simulations are terminated before reaching the optimum value.}
\label{fig:noise-model-two}
\end{figure}

\paragraph{Experimental setup}

We have been using SDPLib, the standard benchmark in SDP. Not all the instances were suitable for the RCP-based solver, which is slow, comparing to classical ones. Out of the full list of 80+ instances, we picked up a few ``good'' ones, according to the following simple criteria:
\begin{itemize}
\item primal problems only.
\item those solved successfully by SCS solver with final objective function close to the reference one (provided by SDPLib).
\item small enough, so as to be  fast to solve:  namely, those which SCS solver managed to solve in less than $200$ seconds.
\end{itemize}

The list of ``good'' instances  appears in the left-most column of the table of results. Note that not every ``good'' instance was actually solved by RCP method because of the size. One can hence notice that the final value of the objective function is sometimes substantially different from the reference one, although the timeout parameter was $86400$ seconds. 
Table~\ref{tab:sdplib2} lists all the results obtained for 1~day timeout ($86400$ seconds). Namely, we interrupt the execution and return solution reached so far as soon as the timeout has been exceeded but simulation is still running. On some instances, the RCP method converged close to the reference value. 
Table~\ref{tab:sdplib-seven-days} lists all similar results obtained for 7~day timeout ($604800$ seconds).
A few instances did not make any substantial progress beyond the very first iteration of the hit{\&}run algorithm because of their large size, such as  \texttt{theta6}, where initial and final objective function values are the same. 


One possibility we have not explored yet, would be to employ sparse matrices for very large problems.

All numerical experiments have been conducted on a server  equipped with 44 cores / 88 hardware threads of \texttt{Intel(R) Xeon(R) CPU E5-2699 v4 @ 2.20GHz}, hundreds of gigabytes of memory, using RedHat 7 OS and Anaconda environment with Python 3.9 and Numpy, Scipy, CVXPY packages installed. The entire code was written in Python.

During our experimentation, we noticed that enabling multi-threading actually significantly slows down the RCP solver. We attribute this behaviour to some software issues, which we did not investigate in depth. Instead, we run all the instances in a separate process without multi-threading, utilising a single CPU core per problem.

\paragraph{Implementation details}

In practice, we faced certain challenges trying to implement RCP as presented. A few modifications in the main algorithms have been made to address those issues:
\begin{enumerate}
\item In the current implementation, the boundary points are computed by the classical generalised eigensolver from Scipy package, which expects a positive definite matrix $B$ (as in Lemma~5), where the problem $A v_i = -\lambda_i B v_i$ is considered. The (imperfect) quantum eigensolver is modelled by adding an artificial noise to the eigenvalues.  In our case, the matrix $A$ is negative (semi)definite, while $B$ is sign indefinite. We actually solve the problem $\mu_i (-A) v_i = B v_i$ and then take the reciprocal $\lambda_i = 1 / \mu_i$. Moreover, when $A$ is close to semi-definiteness, the eigensolver becomes unstable. If that case, we add a tiny value to diagonal elements of $A$ in order to make it strictly positive definite and repeat the eigensolve one more time.
\item In Section~\ref{sec:impl}, the pseudocode of Algorithm \texttt{H\&R\_SDP} in Line~5 performs the isotropization via computation of a square root of matrix $Y$. The same result can be obtained  faster via Cholesky decomposition, which is actually done in our code.
\item In Line~7 of the same algorithm (\texttt{H\&R\_SDP}), a point $y^{(i+1)}$ is generated inside the segment $[\underline{x},\bar{x}]$. In our implementation, we prevent the point from taking end values by a small margin (about $0.001$ of the segment size).
\item In Lines~8 to 10 of the same algorithm (\texttt{H\&R\_SDP}), the loop is repeated a number of times (by default up to $200$ attempts before we claim no further improvements can be done). 
Instead of generating a new couple of boundary points $\underline{x}$, $\bar{x}$ in every iteration of the inner loop, which is very expensive, we repeat Line~7 five times. If $y^{(i+1)}$ remains infeasible,  Line~6 is activated.
\item In the same algorithm (\texttt{H\&R\_SDP}), on Line~3 instead of using the same starting point $x^{(0)}$, we use the last feasible $y$ obtained after previous invocation of \texttt{H\&R\_SDP} algorithm. 
\item The mixing time $M$ is equal to 10 in all simulations. The number $N$ of samples generated in every outer iteration is equal to the  number of variables multiplied by $100$. The latter factor is accountable for the slow performance on large problem instances, although the big number of samples provides better convergence rate.
\end{enumerate}
Sometimes, a slightly deeper minimum can be attained without steps (c) to (e), but this  depends on randomisation. For the hard instances, these steps bring a noticeable improvement of convergence.

\paragraph{Observations}
For the \texttt{qap6} instance, the time spent by the different modules is distributed as follows (as percentage of the total execution time):

\vspace{0.5em}
\begin{tabular}{ll}
Computation of $F(x) = F_0 + \sum_{i=1}^n x_i F_i$: & \qquad $41\,\%$ \\
Computation of generalised eigen-values: & \qquad $22.3\,\%$ \\
Cholesky decomposition: & \qquad $6\,\%$ \\
Generation of random vectors: & \qquad $5.2\,\%$ \\
Other subroutines: & \qquad $25.5\,\%$ \\
\end{tabular}
\vspace{0.5em}

Cholesky decomposition is used to check feasibility of the constraint $F(x) \preceq 0$ and for isotropization. Surprisingly, the computation of $F(x)$  dominates the total run time, although it is implemented very rationally (one line of Python code) with full utilisation of Numpy optimised backend.   

\paragraph{Volume shrinkage}
It might be insightful to see how the volume of convex body is shrinking as iteration process progresses. Again, we selected the \texttt{qap6} instance for  demonstration purposes. This particular instances is challenging (sensitive to sampling scheme), but large enough and solvable in a reasonable time. Fig.~\ref{fig:volume-shrinkage} shows minimum, maximum and mean eigenvalues of the covariance matrix of a cloud of points sampled on every iteration of hit{\&}run algorithm. 

\begin{figure}
\centering
\includegraphics[width=0.85\textwidth]{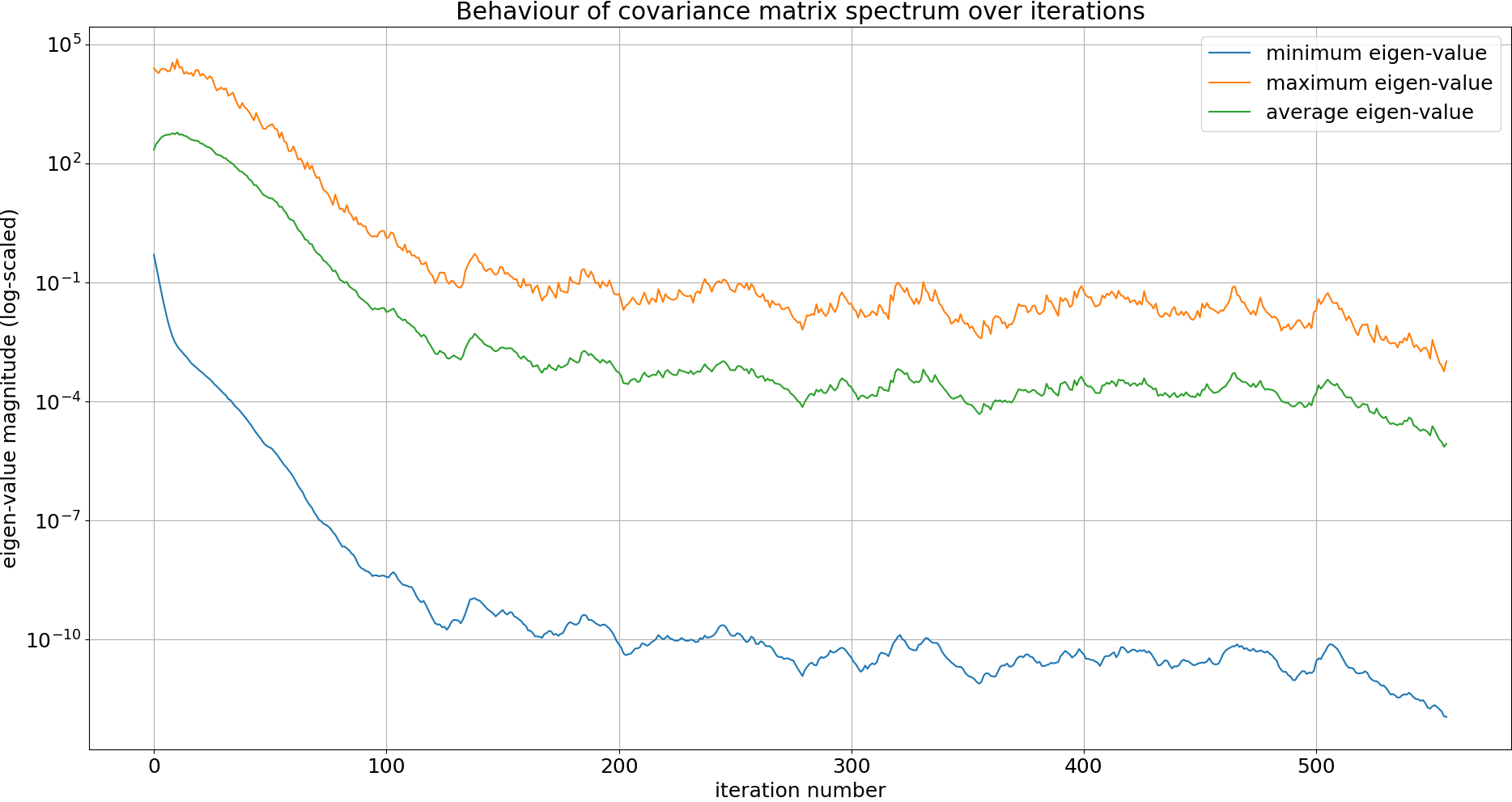}
\caption{Here we demonstrate the process of shrinking of convex body volume over iterations. The volume is roughly proportional to determinant of the covariance matrix of a cloud of sampled points. Since determinant of high-dimensional covariance matrix is either very big or very small, a better visual experience can be drawn from the behaviour of extreme eigen-values as well as the mean one. The data were obtained for \texttt{qap6} problem.}
\label{fig:volume-shrinkage}
\end{figure}

\clearpage
\section{Results of the Implementation}
\label{app:res}
\begin{table*}[h!]
\caption{An overview of the behaviour of RCP on a subset of smaller problems from SDPLIB, a well-known benchmark, within a 24-hour time limit ($86400$ seconds): Instance name, dimensions (first is the number of constraint matrices), reference objective function value, objective function value at termination of RCP, initial objective function at the beginning of RCP, DIMACS Error 1 \cite{borchers1999sdplib,mittelmann2003independent} at termination, DIMACS Error 2 \cite{borchers1999sdplib,mittelmann2003independent} at termination, and run-time in seconds. \\[3mm]}
\label{tab:sdplib2}
\begin{tabularx}{\textwidth}{@{\extracolsep{\fill}}llrrrrrr}
\hline
Instance & size & Ref. & RCP(T) & Initial & Err1(T) & Err2(T) & Time [s] \\
\hline
  gpp100 &  101x100x100 &   -44.94 &  -44.94 &  -18.40 &  950.45 & 0.0 & 86400.32 \\
gpp124-1 &  125x124x124 &    -7.34 &   -6.34 &   18.08 &  152.92 & 0.0 & 86400.30 \\
gpp124-2 &  125x124x124 &   -46.86 &  -45.77 &  -19.86 &  718.02 & 0.0 & 86400.34 \\
gpp124-3 &  125x124x124 &  -153.01 & -150.81 &  -97.13 &  815.00 & 0.0 & 86400.26 \\
gpp124-4 &  125x124x124 &  -418.99 & -407.43 & -267.03 & 1551.43 & 0.0 & 86400.36 \\
   hinf1 &     13x14x14 &     2.03 &    2.09 &    2.25 &    7.30 & 0.0 &    84.31 \\
  hinf10 &     21x18x18 &   108.71 &  122.41 &  151.07 & 4998.56 & 0.0 &   294.37 \\
  mcp100 &  100x100x100 &   226.16 &  226.27 &  318.12 &    0.11 & 0.0 & 86400.23 \\
mcp124-1 &  124x124x124 &   141.99 &  160.95 &  250.15 &    0.08 & 0.0 & 86400.27 \\
mcp124-2 &  124x124x124 &   269.88 &  282.22 &  383.00 &    0.11 & 0.0 & 86400.25 \\
mcp124-3 &  124x124x124 &   467.75 &  486.34 &  614.28 &    0.15 & 0.0 & 86400.25 \\
mcp124-4 &  124x124x124 &   864.41 &  901.95 & 1120.64 &    0.23 & 0.0 & 86400.24 \\
mcp250-2 &  250x250x250 &   531.93 &  741.00 &  758.01 &    0.12 & 0.0 & 86400.73 \\
mcp250-3 &  250x250x250 &   981.17 & 1238.95 & 1272.45 &    0.17 & 0.0 & 86400.63 \\
mcp250-4 &  250x250x250 &  1681.96 & 2100.18 & 2155.74 &    0.26 & 0.0 & 86400.76 \\
    qap5 &    136x26x26 &  -436.00 & -435.98 &  256.19 &   82.88 & 0.0 & 12641.77 \\
    qap6 &    229x37x37 &  -381.44 & -379.29 &  187.39 &   20.77 & 0.0 & 86400.19 \\
    qap7 &    358x50x50 &  -424.82 & -178.42 &  290.45 &   10.08 & 0.0 & 86400.27 \\
    qap8 &    529x65x65 &  -756.96 &  -12.57 &  420.27 &    8.71 & 0.0 & 86400.32 \\
    qap9 &    748x82x82 & -1409.94 &  -38.38 &  745.86 &   10.78 & 0.0 & 86400.40 \\
  theta1 &    104x50x50 &    23.00 &   23.00 &   53.33 &  106.62 & 0.0 &  5967.03 \\
  theta2 &  498x100x100 &    32.88 &  145.51 &  152.09 &  826.16 & 0.0 & 86400.38 \\
  theta3 & 1106x150x150 &    42.17 &   82.34 &   82.34 &  568.98 & 0.0 & 86401.44 \\
  theta4 & 1949x200x200 &    50.32 &  147.05 &  147.05 & 1186.88 & 0.0 & 86403.63 \\
  theta5 & 3028x250x250 &    57.23 &  255.70 &  255.70 & 2327.25 & 0.0 & 86407.97 \\
  theta6 & 4375x300x300 &    63.48 &  410.29 &  410.29 & 4106.66 & 0.0 & 86420.12 \\
  truss1 &      6x13x13 &    -9.00 &   -9.00 &   -7.10 &    5.90 & 0.0 &    17.30 \\
  truss2 &   58x133x133 &  -123.38 & -123.00 &  -21.46 &   75.54 & 0.0 & 86400.21 \\
  truss3 &     27x31x31 &    -9.11 &   -9.11 &   -5.23 &    4.78 & 0.0 &   794.89 \\
  truss4 &     12x19x19 &    -9.01 &   -9.00 &   -5.83 &    5.38 & 0.0 &    89.30 \\
\hline
\end{tabularx}
\end{table*}

\begin{table*}[h!]
\caption{An overview of the behaviour of RCP on a subset of smaller problems from SDPLIB, a well-known benchmark, within a week long time limit ($604800$ seconds): Instance name, dimensions (first is the number of constraint matrices), reference objective function value, objective function value at termination of RCP, initial objective function at the beginning of RCP, DIMACS Error 1 \cite{borchers1999sdplib,mittelmann2003independent} at termination, DIMACS Error 2 \cite{borchers1999sdplib,mittelmann2003independent} at termination, and run-time in seconds. \\[3mm]}
\label{tab:sdplib-seven-days}
\begin{tabularx}{\textwidth}{@{\extracolsep{\fill}}llrrrrrr}
\hline
Instance & size & Ref. & RCP(T) & Initial & Err1(T) & Err2(T) & Time [s] \\
\hline
  gpp100 &  101x100x100 &   -44.94 &  -44.94 &  -18.40  &  950.46 & 0.0 &  82692.28 \\
gpp124-1 &  125x124x124 &    -7.34 &   -7.34 &   18.08  &  212.81 & 0.0 & 205997.33 \\
gpp124-2 &  125x124x124 &   -46.86 &  -46.86 &  -19.86  &  872.42 & 0.0 & 149677.75 \\
gpp124-3 &  125x124x124 &  -153.01 & -153.01 &  -97.13  &  862.31 & 0.0 & 153632.17 \\
gpp124-4 &  125x124x124 &  -418.99 & -418.99 & -267.03  & 1667.97 & 0.0 & 164294.46 \\
   hinf1 &     13x14x14 &     2.03 &    2.09 &    2.25  &    7.30 & 0.0 &     60.85 \\
  hinf10 &     21x18x18 &   108.71 &  122.41 &  151.07  & 4998.56 & 0.0 &    240.51 \\
  mcp100 &  100x100x100 &   226.16 &  226.16 &  318.12  &    0.11 & 0.0 &  74340.28 \\
mcp124-1 &  124x124x124 &   141.99 &  141.99 &  250.15  &    0.07 & 0.0 & 188645.49 \\
mcp124-2 &  124x124x124 &   269.88 &  269.88 &  383.00  &    0.10 & 0.0 & 178055.59 \\
mcp124-3 &  124x124x124 &   467.75 &  467.75 &  614.28  &    0.14 & 0.0 & 171032.22 \\
mcp124-4 &  124x124x124 &   864.41 &  864.41 & 1120.64  &    0.20 & 0.0 & 168804.84 \\
mcp250-2 &  250x250x250 &   531.93 &  579.59 &  758.01  &    0.08 & 0.0 & 604800.46 \\
mcp250-3 &  250x250x250 &   981.17 & 1036.79 & 1272.45  &    0.12 & 0.0 & 604800.48 \\
mcp250-4 &  250x250x250 &  1681.96 & 1775.46 & 2155.74  &    0.17 & 0.0 & 604800.49 \\
    qap5 &    136x26x26 &  -436.00 & -435.98 &  256.19  &   82.88 & 0.0 &  11396.36 \\
    qap6 &    229x37x37 &  -381.44 & -379.67 &  187.39  &   20.75 & 0.0 & 127062.87 \\
    qap7 &    358x50x50 &  -424.82 & -423.03 &  290.45  &   11.47 & 0.0 & 410249.13 \\
    qap8 &    529x65x65 &  -756.96 & -744.99 &  420.27  &   10.04 & 0.0 & 604800.31 \\
    qap9 &    748x82x82 & -1409.94 & -965.46 &  745.86  &   11.90 & 0.0 & 604800.37 \\
  theta1 &    104x50x50 &    23.00 &   23.00 &   53.33  &  106.62 & 0.0 &   5088.24 \\
  theta2 &  498x100x100 &    32.88 &   50.36 &  152.09  &  287.07 & 0.0 & 604800.34 \\
  theta3 & 1106x150x150 &    42.17 &   76.34 &   82.34  &  531.18 & 0.0 & 604800.62 \\
  theta4 & 1949x200x200 &    50.32 &  143.99 &  147.05  & 1188.87 & 0.0 & 604802.40 \\
  theta5 & 3028x250x250 &    57.23 &  255.70 &  255.70  & 2327.25 & 0.0 & 604805.45 \\
  theta6 & 4375x300x300 &    63.48 &  410.29 &  410.29  & 4106.66 & 0.0 & 604811.28 \\
  truss1 &      6x13x13 &    -9.00 &   -9.00 &   -7.10  &    5.90 & 0.0 &     16.35 \\
  truss2 &   58x133x133 &  -123.38 & -123.00 &  -21.46  &   75.54 & 0.0 &  57755.23 \\
  truss3 &     27x31x31 &    -9.11 &   -9.11 &   -5.23  &    4.78 & 0.0 &    684.62 \\
  truss4 &     12x19x19 &    -9.01 &   -9.00 &   -5.83  &    5.38 & 0.0 &     81.22 \\
\hline
\end{tabularx}
\end{table*}

\end{document}